\RequirePackage{amsmath}
\RequirePackage{thm-restate}
\documentclass[runningheads]{llncs}

\usepackage{graphicx}
\usepackage[colorinlistoftodos]{todonotes} 
\usepackage{url}
\usepackage{dsfont}
\usepackage{amssymb}
\usepackage{amsfonts}
\usepackage{mathtools}
\usepackage{hyperref}

\usetikzlibrary{calc}
\usetikzlibrary{shapes.geometric}
\usetikzlibrary{arrows.meta}
\usetikzlibrary{decorations.pathreplacing, decorations.pathmorphing}
\usetikzlibrary{patterns}
\usetikzlibrary{backgrounds}
\usetikzlibrary{scopes}
\usetikzlibrary{arrows, automata, positioning}
\usetikzlibrary{decorations.markings}

\newcommand{\cH}{\mathcal{H}}

\newcommand{\I}{\mathcal{I}}
\newcommand{\X}{\mathcal{X}}

\begin{document}
\title{Node Connectivity Augmentation of Highly Connected Graphs}
%
%
\author{
Waldo Galvez\inst{1}
\and Dylan Hyatt-Denesik\inst{2} 
\and Afrouz Jabal Ameli\inst{2}
\and Laura Sanit\`a\inst{3}}


\authorrunning{W. Galvez et al.}

\institute{Institute of Engineering Sciences, Universidad de O'Higgins, Rancagua, Chile, \email{waldo.galvez@uoh.cl}
\and
Eindhoven University of Technology, Eindhoven, Netherlands, \email{\{d.v.p.hyatt-denesik,a.jabal.ameli\}@tue.nl}
\and
Bocconi University, Milan, Italy, \email{laura.sanita@unibocconi.it}}
\maketitle              
\begin{abstract}
Node-connectivity augmentation is a fundamental network design problem. We are given a $k$-node connected graph $G$ together with an additional set of  links, and the goal is to add a cheap subset of links to $G$ to make it $(k+1)$-node connected. 

In this work, we characterize completely the computational complexity status of the problem, by showing hardness for all values of $k$ which were not addressed previously in the literature. 

We then focus on $k$-node connectivity augmentation for $k=n-4$, which corresponds to the highest value of $k$ for which the problem is NP-hard. We improve over the previously best known approximation bounds for this problem, by developing a $\frac{3}{2}$-approximation algorithm for the weighted setting, and a $\frac{4}{3}$-approximation algorithm for the unweighted setting.    
\end{abstract}

\section{Introduction}

Connectivity augmentation is among the most fundamental problems in network design. Here we are given a $k$-(edge- or node-)connected graph and a set of additional edges (called \emph{links}). The goal is to select the cheapest subset of links to add to the graph, in order to increase its connectivity to $k+1$. 

When dealing with \emph{edge-connectivity}, the problem is quite well understood: $k$-edge-connectivity augmentation is APX-hard for any $k \geq 1$, and in fact, for any given $k$ the problem can be reduced to the case $k=2$ (see~\cite{Dinic76}). In the past years, several constant-factor approximation results have been developed in the literature for edge-connectivity augmentation~\cite{DBLP:journals/talg/Adjiashvili19,DBLP:conf/stoc/Byrka0A20,DBLP:journals/algorithmica/CheriyanG18,DBLP:journals/algorithmica/CheriyanG18a,DBLP:conf/soda/Fiorini0KS18,DBLP:conf/stoc/0001KZ18,nutov20202nodeconnectivity,DBLP:journals/corr/abs-2009-13257,traub2022better,traub2022local}, culminating with the recent breakthroughs in~\cite{DBLP:journals/corr/abs-2012-00086,traub20231}.

When dealing with \emph{node-connectivity}, the problem is instead much less understood, both from a computational complexity and an approximation point of view.
From a complexity perspective, there is no known reduction that allows us to focus on a particular value of $k$ (like e.g. $k=2$ for the edge-connectivity case). 
The problem is known to be APX-hard for $ k \geq 1, k \le n-n^c$ (with $n$ being the number of nodes of the graph and $c$ being any fixed constant)~\cite{KortsarzKL2004}\footnote{Although not stated explicitly in the reference, the reduction requires that $k \le n-n^c$ in order to have polynomial running time. A more detailed discussion in Appendix~\ref{sec:appendix_prevhard}} even if all the link weights are $1$ (known as the \emph{unweighted setting}), but for high values of $k$ the complexity is not completely settled. Bérczi et al.~\cite{BERCZI2012565} observe that the problem is polynomial-time solvable for $k=n-2$ and $k=n-3$, and it is NP-hard for $k=n-4$ when links allow general weights (known as the \emph{weighted setting}). Still, the complexity status is open for $k=n-4$ in the unweighted setting, as well as for any value of $k$ somewhat between $n-4$ and the previously mentioned bound of $n-n^c$.

Also from an approximation perspective, the results are fewer compared to the edge-connectivity case, and mostly concentrate on small values of $k$.
For $k=1$, Frederickson and JáJá~\cite{DBLP:journals/siamcomp/FredericksonJ81} give a 2-approximation. Recently, this bound has been improved in the unweighted case~\cite{angelidakis2022node,Jabal2022,nutov20202nodeconnectivity}.
A $2$-approximation is also known for $k=2$~\cite{auletta199921}, which can be improved in the unweighted setting when the input graph is a cycle~\cite{10.1007/978-3-030-92702-8_1}. For any fixed $k$, the weighted version of the problem admits a $(4+\varepsilon)$-approximation if $n$ is large enough~\cite{DBLP:conf/soda/Nutov20} (see also~\cite{DBLP:journals/siamcomp/CheriyanV14}). For arbitrary values of $k$, it is still open whether a constant approximation is achievable or not, and the best known factor is $O\left(\log\left(\min\left\{\frac{n}{n-k},n-k\right\}\right)\right)$~\cite{Nutov14}. Note that this implies a $O(1)$-approximation for most values of $k$, except for 
{ ``intermediate'' ones like, e.g.,  $n-n^c$ for $0<c<1$.}


\subsection{Our results and techniques}
As a first result, we completely settle the computational complexity status of the problem. In Section~\ref{sec:hardness}, we prove that $(n-d)$-node connectivity augmentation is NP-hard (in fact, APX-hard) for $d \geq 4, d=O(n^c)$, with $c$ being any positive constant, even in the unweighted setting. Our result complements the result of Kortsarz et al.~\cite{KortsarzKL2004}, and hence completes the picture regarding the complexity status of this classical and fundamental problem.


We prove this result by first showing APX-hardness for the case $d=4$ (i.e. $(n-4)$-node-connectivity augmentation), reducing from a variant of SAT denoted as 3-SAT-4, whose optimization version is known to be APX-hard~\cite{berman2003approximation}. 
We then employ an approximation preserving reduction from $(n-d)$-node connectivity augmentation to $(n-(d+1))$-node connectivity augmentation, that holds whenever
$d=O(n^c)$. We remark here that our reduction is different from the hardness construction of~\cite{BERCZI2012565}, as their result heavily relies on the fact that the links have different weights. 

We then focus on the approximability of the $k$-node-connectivity augmentation problem for $k=n-4$ which, as we just discussed, corresponds to the highest value of $k$ for which the problem is NP-hard.   
The best approximation factor known for this problem is a $2$-approximation (which follows e.g. from~\cite{Nutov14}). In Section~\ref{sec:3/2approx}, we improve over this result by giving a $\frac{3}{2}$-approximation for the problem. In Section~\ref{sec:4/3approx}, we improve further the approximability to $\frac{4}{3}$ in the unweighted setting. 

Our approximation algorithms are combinatorial, and rely on the following ingredients. First, we reformulate the problem as a particular \emph{covering} problem in the complement of the graph, as done also in~\cite{BERCZI2012565} (see also~\cite{BercziV10}). In particular, one easily observes that an $(n-d)$-node-connected graph $G$ is \emph{not} $(n-(d-1))$-node-connected if and only if the complement of $G$ has some complete bipartite graphs (which we call \emph{obstructions}) as edge-induced subgraphs. Therefore, the augmentation problem can be reformulated as the problem of selecting links in the complement of $G$ in order to cover all such obstructions (see Definition~\ref{def:dobstruction}). For $d=4$, this covering problem becomes a generalization of the \emph{edge cover} problem, in which one wants to select  links, but in addition to covering all vertices, one also wants to cover all cycles of length $4$ in the graph. For the weighted version of this problem, 
we first compute a suitable partial (infeasible) solution that has at most half of the cost of the optimal one, and then show that we can solve the remaining instance in polynomial time via a reduction to edge cover. For the unweighted version, the approach is somewhat opposite: our algorithm starts by computing a particular edge cover of the instance, and then adjusts the solution to cover the additional obstructions represented by the 4-cycles, while carefully bounding the cost increase. 

As our results show, better constant approximation factors can be obtained (not surprisingly) using ad-hoc techniques for very large values of $k$. We believe that studying the problem for such values can be instrumental to the development of useful tools and techniques, which seems to be needed to understand the approximability of the general node-connectivity augmentation problem.

\subsection{Organization of the paper}

In Section~\ref{sec:prelims}, we reformulate the $(n-d)$-node-connectivity augmentation problem as a \emph{covering} problem that we refer to as $d$-Obstruction Covering, and further investigate the properties of $d$-Obstruction Covering instances for the case $d=4$. In Section~\ref{sec:3/2approx}, we provide our $3/2$-approximation algorithm for the $(n-4)$-Node-Connectivity Augmentation problem, and then in Section~\ref{sec:4/3approx} we provide an improved $4/3$-approximation for the unweighted version of the problem. In Section~\ref{sec:hardness} we prove that the $(n-d)$-Node-Connectivity Augmentation problem is APX-hard for any $d\ge 4$, with $d=O(n^c)$ for any fixed constant $c$. Due to space constraints, several proofs are deferred to the Appendix.

\section{From connectivity augmentation to obstruction covering}
\label{sec:prelims}

As a first step, we reformulate our problem as a covering problem, called the $d$-Obstruction Covering problem, which we define next.
In the following, we denote by $K_{i,j}$ the complete bipartite graph that has $i$ nodes in one bipartition and $j$ nodes in the other bipartition.

\begin{definition}[$d$-Obstruction Covering]\label{def:dobstruction}
    We are given a value $d \le n$, a graph $G=(V,E)$ of $n$ nodes that has no $K_{i,j}$ edge-induced subgraph for any $i,j : i+j >d$, and a subset of links $L \subseteq E$ with associated costs $c(\ell) \geq 0, \forall \ell \in L$. We want to compute a minimum-cost subset of links $F\subseteq L$ such that every $K_{i,j}$ edge-induced subgraph of $G$ with $i+j = d$ contains a link of $F$.
\end{definition}

For a given $d$-Obstruction Covering instance, we call $K_{i,j}$ a \emph{forbidden} subgraph if $i+j > d$, and an \emph{obstruction} if $i+j = d$. An obstruction $K_{i,j}$ is \emph{covered} by link $\ell \in L$ if $\ell$ is an edge of $K_{i,j}$.

As already observed in~\cite{BERCZI2012565,BercziV10}, the $(n-d)$-node-connectivity augmentation problem is equivalent to solving $d$-Obstruction Covering in the complement of the original graph, hence being an equivalent reformulation. See Appendix~\ref{sec:from_connectivity_to_obstructions} for a detailed proof. 
\begin{theorem}
\label{thm:from_connectivity_to_obstructions}
    There is an approximation-preserving reduction from the $(n-d$)-node-connectivity augmentation problem to the $d$-Obstruction Covering problem.
\end{theorem}

Using Theorem~\ref{thm:from_connectivity_to_obstructions}, we can focus on showing a $3/2$-approximation algorithm for the $4$-Obstruction Covering problem.
{ We denote by $(G=(V,E),d,L,c)$ a $d$-Obstruction Covering instance}. {If $c(\ell)=1$ for all $\ell\in L$, then we call this instance \emph{unweighted} and refer to it as $(G,d,L)$ instead. Note that an unweighted instance $(G,d,L)$ is equivalent to instance $(G,d,E,c)$ where $c(e) =1$ for $e\in L$, and $c(e)=\infty$ otherwise. }
We refer to $K_{2,2}$ edge-induced subgraphs of $G$ as \emph{squares}, and to nodes that do not belong to any square as \emph{lonely nodes}.  
{ For any subgraph $C$ of $G$, we denote by $V(C)$ the set of its vertices, and by $E(C)$ the corresponding set of edges}.


\begin{definition}
    A graph is called  a \textit{ladder} if it has maximum degree $3$, and one can label its nodes as $v_1, v_2, \ldots, v_{r}, u_1, u_2, \ldots, u_r$ ($r\ge2$) such that, for every $1\le i < r$, the subgraph induced by the nodes $u_i,$ $v_i,$ $v_{i+1},$ and $u_{i+1}$ is a square with edges $u_iv_i,v_iv_{i+1},v_{i+1}u_{i+1},u_{i+1}u_i$. We refer to the number of squares in this sequence ($r-1$) as the length of the ladder. 
    A graph is called a \textit{hexagon} if it has a spanning edge-induced subgraph that is isomorphic to the graph in Figure~\ref{fig:hexagon}.(b).
\end{definition}
\begin{figure}[t!]
     \centering
     (a)
         \centering
         \begin{tikzpicture}[scale=0.7]
            \tikzset{black dot/.style={draw=black, very thick, circle,minimum size=0pt, inner sep=1pt, outer sep=1pt,fill=black}}
            \tikzset{terminal/.style={draw=black,  thick,minimum size=0pt, inner sep=2.5pt, outer sep=1pt}}
            \tikzset{P node/.style={fill={rgb,255: red,20; green,154; blue,0}, draw={rgb,255: red,20; green,154; blue,0}, circle, minimum size=0pt,inner sep=1pt, outer sep=1pt}}
            \tikzset{Writing/.style={shape=circle} }
        
            \tikzstyle{witness edge}=[-, draw={rgb,255: red,195; green,0; blue,3}, very thick]
            \tikzstyle{T edges}=[-, very thick]
            \tikzstyle{new witness}=[-, draw={rgb,255: red,195; green,0; blue,3}, dashed, very thick]
            \tikzstyle{connected terminals}=[-, draw=black, dashed, very thick]
            \tikzstyle{P}=[-, draw={rgb,255: red,20; green,154; blue,0}, very thick]
                
            \node [style=black dot] (0) at (-2, 1) {};
            \node [style=black dot] (1) at (-2, -0.5) {};
            \node [style=black dot] (2) at (1, -0.5) {};
            \node [style=black dot] (3) at (1, 1) {};
            \node [style=black dot] (4) at (-5, 1) {};
            \node [style=black dot] (5) at (-3.5, 1) {};
            \node [style=black dot] (6) at (-3.5, -0.5) {};
            \node [style=black dot] (7) at (-0.5, -0.5) {};
            \node [style=black dot] (8) at (-0.5, 1) {};
            \node [style=black dot] (9) at (-5, -0.5) {};
            \node [style=black dot] (10) at (2.5, -0.5) {};
            \node [style=black dot] (11) at (2.5, 1) {};
            \node (12) at (-4.75, 1.25) {$v_1$};
            \node (13) at (-3.25, 1.25) {$v_2$};
            \node (14) at (-1.75, 1.25) {$v_3$};
            \node (15) at (-0.25, 1.25) {$v_4$};
            \node (16) at (1.25, 1.25) {$v_5$};
            \node (17) at (2.75, 1.25) {$v_6$};
            \node (18) at (-4.75, -0.25) {$u_1$};
            \node (19) at (-3.25, -0.25) {$u_2$};
            \node (20) at (-1.75, -0.25) {$u_3$};
            \node (21) at (-0.25, -0.25) {$u_4$};
            \node (22) at (1.25, -0.25) {$u_5$};
            \node (23) at (2.75, -0.25) {$u_6$};
        
            \draw [style=T edges] (4) to (9);
            \draw [style=T edges] (9) to (10);
            \draw [style=T edges] (10) to (11);
            \draw [style=T edges] (11) to (4);
            \draw [style=T edges] (5) to (6);
            \draw [style=T edges] (1) to (0);
            \draw [style=T edges] (8) to (7);
            \draw [style=T edges] (2) to (3);
            \draw [style=T edges, bend right=15] (9) to (10);
        \end{tikzpicture} 
         (b)
        \centering
        \begin{tikzpicture}[scale=0.7]
            \tikzset{black dot/.style={draw=black, very thick, circle,minimum size=0pt, inner sep=1pt, outer sep=1pt,fill=black}}
            \tikzset{terminal/.style={draw=black,  thick,minimum size=0pt, inner sep=2.5pt, outer sep=1pt}}
            \tikzset{P node/.style={fill={rgb,255: red,20; green,154; blue,0}, draw={rgb,255: red,20; green,154; blue,0}, circle, minimum size=0pt,inner sep=1pt, outer sep=1pt}}
        
            \tikzstyle{witness edge}=[-, draw={rgb,255: red,195; green,0; blue,3}, very thick]
            \tikzstyle{T edges}=[-, very thick]
            \tikzstyle{new witness}=[-, draw={rgb,255: red,195; green,0; blue,3}, dashed, very thick]
            \tikzstyle{connected terminals}=[-, draw=black, dashed, very thick]
            \tikzstyle{P}=[-, draw={rgb,255: red,20; green,154; blue,0}, very thick]
            
            \node [style=black dot] (0) at (0, 0) {};
            \node [style=black dot] (1) at (0, 1.25) {};
            \node [style=black dot] (2) at (-1.25, -0.75) {};
            \node [style=black dot] (3) at (1.25, -0.75) {};
            \node [style=black dot] (4) at (1.25, 0.75) {};
            \node [style=black dot] (5) at (-1.25, 0.75) {};
            \node [style=black dot] (6) at (0, -1.25) {};
            \node (7) at (0.25, 1.5) {$v_1$};
            \node (8) at (1.5, 0.75) {$v_2$};
            \node (9) at (0.25, 0.2) {$v_0$};
            \node (10) at (-1.5, 0.75) {$v_6$};
            \node (11) at (-1.5, -0.75) {$v_5$};
            \node (12) at (0.25, -1.5) {$v_4$};
            \node (13) at (1.5, -0.75) {$v_3$};
            \draw [style=T edges] (1) to (0);
            \draw [style=T edges] (0) to (2);
            \draw [style=T edges] (0) to (3);
            \draw [style=T edges] (4) to (3);
            \draw [style=T edges] (4) to (1);
            \draw [style=T edges] (1) to (5);
            \draw [style=T edges] (5) to (2);
            \draw [style=T edges] (2) to (6);
            \draw [style=T edges] (6) to (3);
        \end{tikzpicture}
     \caption{(a) An example of a ladder $C$, with $r=6$.  
     (b) An example of a hexagon.}
     \label{fig:hexagon}
\end{figure}
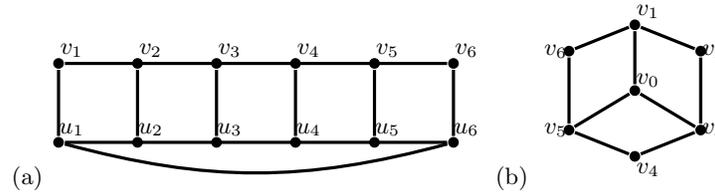

For a given ladder, a labelling of the nodes according to the above definition is called a \emph{consistent} labelling. 

Let $C$ be a hexagon or a ladder, that is an edge-induced subgraph of a given graph $G=(V,E)$. We will let $\delta(V(C))$ denote the set of edges between $V(C)$ and $V\setminus V(C)$, and to simplify notation, we will let  $\delta(C)\coloneqq \delta(V(C))$.
We say that the degree of $C$ is $|\delta(C)|$ (i.e., the number of edges between $V(C)$ and $V\setminus V(C)$). 
We say that a node $v$ is a \emph{corner} of $C$, if $v$ is incident to an edge in $\delta(C)$.
Observe that if $G$ is a graph of a $4$-Obstruction Covering instance, then a corner $v$ has degree $3$ because $K_{1,4}$ is a forbidden subgraph. In this case, note that a hexagon has at most three corners and a ladder has at most four corners. 
From now on, whenever we use the word subgraph, we refer to edge-induced subgraph (unless specified otherwise).

Let $(G=(V,E),4,L,c)$ be a $4$-Obstruction Covering instance.
The following theorem is a useful structural result, which states that we can decompose $G$ into node-disjoint hexagons, lonely nodes, and ladders of maximal length. Its proof can be found in Appendix~\ref{sec:ChainDecomposition}.
\begin{theorem}\label{thm:ChainDecomposition}
     We can find in polynomial time node-disjoint subgraphs  $R, G_1, G_2,$ $ \ldots, G_k$ of $G$ such that:
    \begin{itemize}
        \item Each $G_i$ is a hexagon or a ladder of maximal length. If $G_i$ is a ladder, then we also have a consistent labelling. 
        \item $V(R) = V\backslash \left( V(G_1)\cup V(G_2) \ldots \cup V(G_k) \right)$ is a set of lonely nodes.
        \item For every square subgraph $S$ of $G$, there exists an index $i$ such that $V(S)\subseteq V(G_i)$.
    \end{itemize}
\end{theorem}

The decomposition provided by Theorem~\ref{thm:ChainDecomposition} is useful for our purposes, because it is possible to develop a Dynamic Programming based algorithm that solves the $4$-Obstruction Covering problem when restricted to a ladder or a hexagon (in fact, something slightly more general than that). The following theorem, whose proof is deferred to Appendix~\ref{sec:DP}, provides a precise statement of which instances can be solved.
\begin{restatable}{theorem}{DP}
    \label{alg:DP}
    Let $C=(V,E)$ be a ladder with a given consistent labelling or a hexagon.  
    Given edge weights $w:E\rightarrow \mathbb{R}$ and node covering requirements $h: V \rightarrow \{0,1\}$, we can find in polynomial time a minimum cost set $F\subseteq E$ such that, for every square $S$ in $C$, $F\cap E[S] \neq \emptyset$, and for every $v\in V$ with $h(v)=1$, at least one edge of $F$ is incident to $v$. 
\end{restatable}
In particular, we can use Theorem~\ref{alg:DP} to solve each ladder or hexagon of degree zero optimally.
\begin{corollary}
\label{cor:degree0}
    Consider a $4$-Obstruction Covering instance $(G,4,L,c)$, and a given decomposition of $G$ into lonely nodes $R$, and disjoint hexagon or ladder subgraphs $G_1, \dots, G_k$ (as in Theorem~\ref{thm:ChainDecomposition}). For any $G_i$ of degree $0$,  one can find in polynomial time an optimal covering of every obstruction in $G_i$.
\end{corollary}
Thanks to Corollary~\ref{cor:degree0}, we can assume that $G$ contains no hexagons or ladders of degree $0$. Moreover,  we assume that $G$ is connected, as we can apply our algorithm to each connected component individually.

\section{ A \texorpdfstring{$\frac{3}{2}$}{3/2}-approximation for Weighted \texorpdfstring{$(n-4)$}{(n-4)}-node connectivity augmentation}
\label{sec:3/2approx}
In this section, we will prove the following theorem.
\begin{theorem}
\label{thm:3/2-apx}
    There is a polynomial time $3/2$-approximation algorithm for the $(n-4)$-Node-Connectivity Augmentation problem. 
\end{theorem}

We first give a high-level description of the algorithm. By Theorem~\ref{thm:from_connectivity_to_obstructions} we focus on $4$-Obstruction Covering instances. For a $4$-Obstruction Covering instance $(G,4,L,c)$, we consider a fixed optimal solution $opt$, with cost $c(opt)$.
We apply Theorem~\ref{thm:ChainDecomposition} to $G$ to find node-disjoint hexagons and ladders $G_1,\dots, G_k$, and lonely nodes $R$. 
Our algorithm first finds a partial solution $F_1$, with cost $c(F_1)\leq \frac{1}{2}c(opt)$, that covers at least the corners of  hexagons and ladders $G_1,\dots, G_k$ whose degree is at least three. We then remove $F_1$ from $E$, leaving an instance with only ladders and hexagons of degree $1$ or $2$. We then compute a special edge cover $F_2$, to take care of all the remaining obstructions, with cost $c(F_2) \leq c(opt)$. Putting $F_1$ and $F_2$ together then defines a feasible solution of cost at most $\frac{3}{2}c(opt)$. We now give more details about the computation of these partial solutions.

To compute the partial solution $F_1$, we make use of the following useful Corollary of Petersen’s Theorem\cite{Pet1981}.
It is proved in Appendix~\ref{sec:4regularFactorization}.
\begin{corollary}\label{cor:4regularFactorization}
    Let $G$ be a $4$-regular graph. Then, we can decompose $G$ into two $2$-regular spanning subgraphs in polynomial time.
\end{corollary}

We also make use of a slight generalization of the Edge Cover problem, denoted as Minimum $N$-Edge Cover problem.
Given a graph $G=(V,E)$, a subset $N\subseteq V$, and an edge cost function $c:E \rightarrow \mathbb{R}$, we want to find a minimum cost subset $E' \subseteq E$, such that the endpoints of $E'$ contain $N$. 
The following lemma follows easily from the fact that the Edge Cover problem is polynomial-time solvable, but for the sake of completeness its proof can be found in Appendix~\ref{sec:SEdgeCover}.
\begin{restatable}{lemma}{SEdgeCover}
\label{lem:SEdgeCover}
The Minimum $N$-Edge Cover problem can be solved efficiently.
\end{restatable}

We now consider node-disjoint hexagons and ladders $G_1,$ $\dots, G_k$, and lonely nodes $R$ obtained from Theorem~\ref{thm:ChainDecomposition}.
Denote the corners of ladders and hexagons of degree at least 3 by $N$, that is, each node of a ladder or hexagon that is adjacent to a node outside the ladder or hexagon containing it is in $N$.
By Lemma~\ref{lem:SEdgeCover}, we compute a minimum cost $N$-edge cover, that we denote by $E_C$. As nodes of $N$ (together with their neighbours) induce a subset of the obstructions that must be covered (they are nodes of degree $3$, hence each with its neighbours induces $K_{1,3}$), it is clear that $c(E_C) \leq c(opt)$. We then construct a 4-regular graph $G'$ using $E_C$, and apply Corollary~\ref{cor:4regularFactorization} to find a subset of edges of cost $\leq \frac{1}{2}c(opt)$.

We begin with $G'= (V' = \emptyset,E' = \emptyset)$, and costs $c' : E' \rightarrow \mathbb{R}$. For every $G_i$, $i=1,\dots, k$ of degree at least $3$ in $G$, add corresponding node $g_i$ to $V'$. 

We next add edges to $E'$: consider the links of $E_C$ in an arbitrary but fixed order, say $\ell_1,...,\ell_{|E_C|}$. For any endpoint $u$ of $\ell_i$ that is not incident to any $\ell_j$ for $j<i$, we say that $u$ is \emph{satisfied} by $\ell_i$. Note, each corner is satisfied by exactly one link $\ell_i$.
For link $\ell_i = uv \in E_C$, suppose  $\ell_i$ has two endpoints in $N$ in (possibly equal) subgraphs $G_a$ and $G_b$, $a,b\in \{1,\dots, k\}$. 
If both $u$ and $v$ are satisfied by $\ell_i$, then we add an edge $g_ag_b$, with label $\ell_i$, to $E'$ with cost $c'(g_ag_b) = c(\ell_i)$ (adding a loop in $G'$ with cost $c(\ell_i)$ if $G_a = G_b$). If exactly one endpoint $u \in \ell_i$ is not satisfied by $\ell_i$, say $u \in G_a$, then we add a dummy node $\bar u$ to $V'$ and add edge $\bar u g_b$, with label $\ell_i$, to $E'$ with cost $c'(\bar u g_b) = c(\ell_i)$. If both endpoints of $\ell_i$ are satisfied already, then we remove $\ell_i$ from $E_C$ for no loss of feasibility and no increase in cost as the edge costs are non-negative.

Now consider a link $\ell_i=uv \in E_C$ that has exactly one endpoint in $N\cap G_a$ for some $a\in \{1,\dots, k\}$, say $v\in G_a$. 
If $v$ is satisfied by $\ell_i$, we add a dummy node $\bar u$ to $V'$ and add an edge $\bar u g_a$, with label $\ell_i$, and cost $c'(\bar u g_a) = c(\ell_i)$ to $E'$. If this $v$ is already satisfied, then we can simply remove $\ell_i$ from $E_C$. 
Notice that, if $G'$ is not $4$-regular, then we can add a polynomial number of dummy nodes to $V'$ and zero cost dummy edges to $E'$ such that $G'$ becomes $4$-regular 
(see Appendix~\ref{sec:deg4} for more details).

By applying Corollary~\ref{cor:4regularFactorization}, we find two edge disjoint $2$-regular spanning subgraphs $H_1$ and $H_2$ such that $H_1\cup H_2 = G'$ and $\min \{c(E(H_1)), c(E(H_2)) \} \leq \frac{1}{2}c(opt)$. The minimum cost set $H\in \{ H_1, H_2\}$ induces our choice of $F_1 \subseteq L$. 

In order to compute the second part of the solution, denoted $F_2$, we use $F_1$ to find a $4$-Obstruction Covering instance $(G''=(V'',E''), 4, L'' ,c'')$ with cost $c'': E'' \rightarrow \mathbb{R}$, such that the subgraphs $G_1', \dots, G_k'$ found by Theorem~\ref{thm:ChainDecomposition} have degree at most $2$. We begin with $G'' = G$.    
For every $f \in F_1$, if $f\in \delta(G_i)$, for some $i\in \{1,\dots, k\}$, we remove $f$ from $E''$. If $f\in E[G_i]$ for some $i\in \{1,\dots, k\}$, then $f$ satisfies at least one node in $N\cap V(G_i)$.  For such a link $f\in E[G_i]$ with endpoint $u$, 
since $u\in N\cap V(G_i)$, $u$ has degree $3$ and must be incident to an edge $uv\in \delta(G_i)$. We add a dummy node $\bar u$ to $G''$ and remove $uv$ from $G''$ and replace it with $\bar u v$, giving it label $uv$, with cost $c''(\bar u v) = c'(uv)$. Finally, we set $c''(f) = 0$. We prove the following lemma in Appendix~\ref{sec:reducedegrees}
\begin{lemma}
\label{lem:reducedegrees}
    $(G'',4,L'',c'')$ is a $4$-Obstruction Covering instance such that: \begin{enumerate} \item The subgraphs  $G_1' ,\dots, G_k'$ found by Theorem~\ref{thm:ChainDecomposition} each have degree $1$ or $2$; \item The cost of the optimal solution to $(G'',4,L'',c'')$ is at most $c(opt)$; and \item Given a solution $F$ to 
    $(G'',4,L'',c'')$ of cost $c''(F)$, we can find in polynomial time a solution to $(G,4,L,c)$ of cost $c''(F)+c(F_1)$.\end{enumerate}
\end{lemma}

The last ingredient is the following lemma, proven in Appendix~\ref{sec:degree12edgecover}, which essentially shows that the instance  $(G'',4,L'',c'')$ is solvable in polynomial time.
\begin{lemma}
\label{lem:degree12edgecover}
    Given a $4$-Obstruction Covering instance $(G,4,L,c)$, if the subgraphs  $G_1,\dots, G_k$  found as in Theorem~\ref{thm:ChainDecomposition} have degree $1$ or $2$, then we can find an optimal solution in polynomial time.
\end{lemma}

We now conclude putting everything together to prove Theorem~\ref{thm:3/2-apx}.

\begin{proof}[Proof of Theorem~\ref{thm:3/2-apx}]
    Given an $(n-4)$-Node Connectivity Augmentation instance, we apply Theorem~\ref{thm:from_connectivity_to_obstructions} to obtain an instance $(G,4,L,c)$ of $4$-Obstruction Covering. By Theorem~\ref{thm:ChainDecomposition}, we obtain node-disjoint ladders and hexagons $G_1,\ldots,$ $G_k$, and lonely nodes $R$. Using the discussion above, we can compute in polynomial time a partial solution $F_1$ such that $c(F_1)\le \frac{1}{2}c(opt)$, to cover a subset of the obstructions (namely, corners of ladders and hexagons of degree at least three in this decomposition). We then construct a new instance $(G'',4,L'',c'')$,  in which the node-disjoint ladders and hexagons found by Theorem~\ref{thm:ChainDecomposition} have degree at most two. By applying Lemma~\ref{lem:degree12edgecover} to $(G'',4,L'',c'')$, we can find an optimal solution $F_2$ for this instance. Then, by applying Lemma~\ref{lem:reducedegrees}, we can find a solution of cost $c(F_1)+c(F_2)\le \frac{3}{2}c(opt)$ for $(G,4,L,c)$. 
\end{proof}

\section{ A $\frac{4}{3}$-approximation for unweighted $(n-4)$-node-connectivity augmentation}
\label{sec:4/3approx}
In this section, we will prove the following theorem.
\begin{theorem}
\label{thm:unweighted}
    There is a polynomial time $\frac{4}{3}$-approximation algorithm for the unweighted $(n-4)$-Node Connectivity Augmentation problem.
\end{theorem}
We first consider an unweighted $4$-Obstruction Covering instance $(G,4,L)$, and apply Theorem~\ref{thm:ChainDecomposition} to decompose $G$ into lonely nodes $R$, and node-disjoint hexagons, ladders $G_1,\dots, G_k$. We fix an optimal solution, denoted $opt$.

Our algorithm will apply some preprocessing operations, and then find a minimum size subset of edges $EC_3\subseteq L$, that covers every degree $3$ node in $G$ and the obstructions of every  $G_i\in\{G_1,\dots, G_k\}$ of degree 1 and 2.
As the obstructions covered by $EC_3$ are a subset of all obstructions,  we can see $c(EC_3) \leq c(opt)$. We then show that we can modify this solution to get a solution $APX$ that covers the remaining obstructions, and by a careful local charging argument we will show that $c(APX) \leq \frac{4}{3}c(opt)$. 

To simplify our arguments and allow us to focus on  discussing the main technical challenges of our approach, 
we will assume that the subgraphs $G_1,\dots, G_k$ given by Theorem~\ref{thm:ChainDecomposition} consist only of ladders of length $1$. Indeed, it is possible to prove that, if we aim to get a $4/3$-approximation for the general problem, we can restrict ourselves to this special case (see Appendix~\ref{sec:4/3approx-proofs}).

\smallskip
\emph{Preprocessing and Computing $EC_3$}. We call a link $\ell$ \emph{necessary} if there is an obstruction $O$ such that $\ell$ is the only link that can cover $O$. Observe  that any feasible solution must contain all necessary links. Therefore, we initially take these links as part of our solution and then remove them from $E$ and $L$ and consider the remaining instance. Thus, we assume without loss of generality that no link in $L$ is necessary.

The following lemma, similar in spirit to Lemma~\ref{lem:degree12edgecover},  allows us to compute $EC_3$ in polynomial time (Proof found in Appendix~\ref{sec:degree123}).
\begin{lemma}
\label{lem:degree123}
    Given an unweighted $4$-Obstruction Covering instance $(G,4,L)$, let $R, G_1,\dots, G_k$  be the subgraphs  found as in Theorem~\ref{thm:ChainDecomposition}. We can find in polynomial time a minimum cardinality subset of links that:
    \begin{enumerate}
        \item covers every node of degree 3, and
        \item covers every square $S$ with $V(S)\subseteq V(G_i)$, for some $G_i\in \{G_1,\dots, G_k\}$ that has degree 1 or 2.
    \end{enumerate}
\end{lemma}
\emph{Fixing the solution}.
Given a subset $H\subseteq \{G_1,\dots, G_k\}$, we denote by $G[H]$ the subgraph induced by $\cup_{G_a\in H} V(G_a)$, and denote by $E[H]$ the edges of $G[H]$. The following lemma shows that if every $G_i\in \{G_1,\dots, G_k\}$ has degree $3$ or $4$, then we can use $EC_3$ to find a feasible solution $APX$ such that $|APX|\leq \frac{4}{3}EC_3$.
\begin{lemma}
\label{lem:simplifiedapproximation}
    Consider an unweighted $4$-Obstruction Covering instance $(G,4,L)$, with decomposition into subgraphs $R, G_1,\dots, G_k$, found by Theorem~\ref{thm:ChainDecomposition},  such that every $G_i\in \{G_1,\dots, G_k\}$ is a ladder of length 1 that has degree $3$ or $4$. 
    Consider a partial solution $EC_3\subseteq L$ that covers every degree $3$ node, but does not (necessarily) cover the squares of $G_i\in \{G_1,\dots, G_k\}$.
    
    We can find in polynomial time, a feasible solution $APX\subseteq L$, such that $|APX| \leq \frac{4}{3}|EC_3|$.
\end{lemma}
\begin{proof}[Sketch of Proof]
    We begin by defining notation. For a square $S$, we call a corner $u\in V(S)$ a \emph{good} node if there is another corner $v\in V(S)$ such that $uv\in L$, and we say $u$ is a bad node otherwise (recall, we say a node of $V(S)$ is a corner if it is incident to an edge in $\delta(S)$). Note that the nodes of  $V(S)$ that are not corners are neither good nor bad. 
    Note that since no links are necessary, all nodes of a degree 4 ladder are good nodes. 

    We find $APX$ by performing a sequence of steps. We begin with step 1 where we set $APX_1 = EC_3$, $Y_1= \emptyset$, and $\cH_1$ equal to the subset of $\{G_1,\dots, G_k\}$ whose squares are covered by $APX_1$. For each step $i$, $i \ge 2$, we define
    \begin{itemize}
    \item  $APX_i\subseteq L$ to be a partial solution;
     \item $\cH_i\subseteq \{G_1,\dots, G_k\}$ to be the set of ladders whose squares are covered by $APX_i$; 
     \item $Y_i$ to be the subset of edges  $ \subseteq APX_i \cap \delta(G[\cH_i])$ such that every $y\in Y_i$ has an endpoint in $R$.
    \end{itemize}
    Our goal is to compute $APX_i$ at $i$-th step, $i\ge 2$, such that the following holds: 
    \begin{enumerate}
        \item[(1)] The obstructions covered by $APX_{i-1}$ are a strict subset of those covered by $APX_i$;
        \item[(2)] $|APX_i\cap (Y_i \cup E[\cH_i])| \leq \frac{4}{3}| EC_3\cap (Y_i \cup E[\cH_i])|$, and $APX_i \backslash (Y_i \cup E[\cH_i]) = EC_3 \backslash (Y_i \cup E[\cH_i])$. 
    \end{enumerate}
    Note that since $APX_i$ covers the obstructions covered by $APX_{i-1}$, we must have $\cH_{i-1}\subseteq \cH_i$. If in iteration $i$, we have $\cH_i = \{G_1,\dots, G_k\}$, then we are done as $APX_i$ covers every obstruction. If $\cH_i \subsetneq \{G_1,\dots, G_k\}$, then we will find some $G_s\in \{G_1,\dots, G_k\} \backslash \cH_i$, and using $G_s$, we will find $APX_{i+1}$, $\cH_{i+1}$, and $Y_{i+1}$.
    
    We will consider three cases for $G_s$, for which we now provide a brief description. First, if $G_s$ has two nodes that are adjacent to  nodes in $R$ or in $\cH_i$. Second, if a good node of $G_s$ is adjacent to a good node of some $G_t\notin \cH_i$. And finally, if neither of the previous cases hold, then we show $G_s$ is degree 3 and that a good node of $G_s$ must be adjacent to a bad node of some degree 3 ladder $G_t\notin \cH_i$. In all cases, we will be able to find at least three links in $\delta(G_s)$ (and possibly $\delta(G_t)$), such that, by adding links to cover $G_s$ (and possibly $G_t$), we can charge the addition of these new links to the links in $\delta(G_s)$ (and possibly $\delta(G_t)$) for a $\frac{4}{3}$ fractional increase. 
    To account for space limitations, we only show the details of case 1 here. All cases can be found in Appendix~\ref{sec:simplifiedapproximationproof}. 

    \textit{Case 1:} there exists $G_s \in \{G_1,\dots, G_k\}\backslash \cH_i$, such that $G_s$ has two nodes that are adjacent to either a node in $R$ or in $\cH_i$.
    Since $G_s$ is degree at least $3$, there are links $\{\ell,\ell_1,\ell_2\} \subseteq \delta(G_s)\cap APX_i$. Without loss of generality, let the links $\ell_1$ and $\ell_2$ be the ones incident to nodes in $\cH_i$ or in $R$.
    
    Since no links are necessary, there is a link $e\in E(G_s)$ that shares an endpoint with $\ell$, and we start with $APX_{i+1}\leftarrow APX_{i}\cup \{e\}$. 
    Note that this ensures that $G_s$ is in $\cH_{i+1}$, and so $\cH_i \subsetneq \cH_{i+1}$.
    
    If any of $\ell_1$, $\ell_2$ or $\ell$ are incident to a node in $R$, then that link is in $Y_{i+1}$. Assume $\ell$ is incident to some ladder $G_t$. If $G_t\in \cH_i$, we are done this iteration. If not, we are going to change $APX_{i+1}$ by swapping two edges. Since $\ell$ is not necessary, there is link $f\in E(G_t)$ sharing an endpoint with $\ell$. We replace $\ell$ with $f$ in $APX_{i+1}$, and so also $G_t$ will be in $H_{i+1}$. 
    Since the endpoints of $\ell$ are covered in $APX_{i+1}$ by either $\ell$ or $\{e,f\}$, the obstructions in $G_s$ are covered, and all other obstructions of $APX_i$ are covered by $APX_{i+1}$, $(1)$ holds (see Figure~\ref{fig:case1sketch}).
    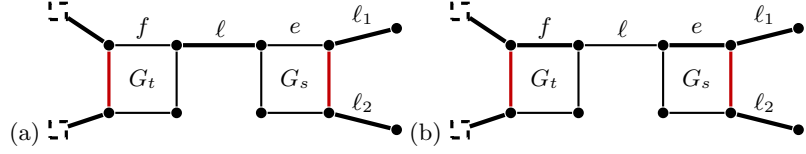
\begin{figure}[t]
    \begin{center}
        \begin{tabular}{c c}
            (a)
            \begin{tikzpicture}[scale=0.9]
                
                \tikzset{black dot/.style={draw=black, very thick, circle,minimum size=0pt, inner sep=1pt, outer sep=1pt,fill=black}}
                \tikzset{terminal/.style={draw=black,  thick,minimum size=0pt, inner sep=2.5pt, outer sep=1pt}}
                \tikzset{P node/.style={fill={rgb,255: red,20; green,154; blue,0}, draw={rgb,255: red,20; green,154; blue,0}, circle, minimum size=0pt,inner sep=1pt, outer sep=1pt}}
            
                \tikzstyle{witness edge}=[-, draw={rgb,255: red,195; green,0; blue,3}, very thick]
                \tikzstyle{T edges}=[-, thick]
                \tikzstyle{Fat edge}=[-, ultra thick]
                \tikzstyle{new witness}=[-, draw={rgb,255: red,195; green,0; blue,3}, dashed, ultra thick]
                \tikzstyle{connected terminals}=[-, draw=black, dashed, very thick]
                \tikzstyle{P}=[-, draw={rgb,255: red,20; green,154; blue,0}, very thick]
                
    		\node [style=black dot] (1) at (0, 0) {};
    		\node [style=black dot] (2) at (0, 1) {};
    		\node [style=black dot] (3) at (1, 1) {};
    		\node [style=black dot] (4) at (1, 0) {};
    		\node [style=black dot] (5) at (2, 1.25) {};
    		\node [style=black dot] (6) at (2, -0.25) {};
    		\node [style=black dot] (7) at (-1.25, 1) {};
    		\node [style=black dot] (8) at (-1.25, 0) {};
    		\node [style=black dot] (9) at (-2.25, 0) {};
    		\node [style=black dot] (10) at (-2.25, 1) {};
    		\node [style=connected terminals] (11) at (-3, 1.5) {};
    		\node [style=connected terminals] (12) at (-3, -0.25) {};
                \node (13) at (0.5, 0.5) {$G_s$};
        		\node (14) at (-1.75, 0.5) {$G_t$};
    		\node (16) at (-0.6, 1.25) {$\ell$};
                \node (16) at (0.5, 1.25) {$e$};
                \node (16) at (-1.75, 1.25) {$f$};
                  \node (16) at (1.5, 1.45) {$\ell_1$};
                  \node (16) at (1.5, 0.15) {$\ell_2$};
    		\node (17) at (-1, 1.25) {};
    		\draw [style=T edges] (4) to (1);
    		\draw [style=T edges] (1) to (2);
    		\draw [style=T edges] (7) to (8);
    		\draw [style=T edges] (8) to (9);
    		\draw [style=T edges] (10) to (7);
    		\draw [style=witness edge] (9) to (10);
    		\draw [style=witness edge] (3) to (4);
    		\draw [style=T edges] (2) to (3);
    		\draw [style=Fat edge] (4) to (6);
    		\draw [style=Fat edge] (3) to (5);
    		\draw [style=Fat edge] (2) to (7);
    		\draw [style=Fat edge] (10) to (11);
    		\draw [style=Fat edge] (9) to (12);
            \end{tikzpicture}
            &
            (b)
            \begin{tikzpicture}[scale=0.9]
                
                \tikzset{black dot/.style={draw=black, very thick, circle,minimum size=0pt, inner sep=1pt, outer sep=1pt,fill=black}}
                \tikzset{terminal/.style={draw=black,  thick,minimum size=0pt, inner sep=2.5pt, outer sep=1pt}}
                \tikzset{P node/.style={fill={rgb,255: red,20; green,154; blue,0}, draw={rgb,255: red,20; green,154; blue,0}, circle, minimum size=0pt,inner sep=1pt, outer sep=1pt}}
            
                \tikzstyle{witness edge}=[-, draw={rgb,255: red,195; green,0; blue,3}, very thick]
                \tikzstyle{T edges}=[-, thick]
                \tikzstyle{Fat edge}=[-, ultra thick]
                \tikzstyle{new witness}=[-, draw={rgb,255: red,195; green,0; blue,3}, dashed, very thick]
                \tikzstyle{connected terminals}=[-, draw=black, dashed, very thick]
                \tikzstyle{P}=[-, draw={rgb,255: red,20; green,154; blue,0}, very thick]
                
    		\node [style=black dot] (1) at (0, 0) {};
    		\node [style=black dot] (2) at (0, 1) {};
    		\node [style=black dot] (3) at (1, 1) {};
    		\node [style=black dot] (4) at (1, 0) {};
    		\node [style=black dot] (5) at (2, 1.25) {};
    		\node [style=black dot] (6) at (2, -0.25) {};
    		\node [style=black dot] (7) at (-1.25, 1) {};
    		\node [style=black dot] (8) at (-1.25, 0) {};
    		\node [style=black dot] (9) at (-2.25, 0) {};
    		\node [style=black dot] (10) at (-2.25, 1) {};
    		\node [style=connected terminals] (11) at (-3, 1.5) {};
    		\node [style=connected terminals] (12) at (-3, -0.25) {};
                \node (13) at (0.5, 0.5) {$G_s$};
        		\node (14) at (-1.75, 0.5) {$G_t$};
    		\node (16) at (-0.6, 1.25) {$\ell$};
                  \node (16) at (0.5, 1.25) {$e$};
                  \node (16) at (-1.75, 1.25) {$f$};
                  \node (16) at (1.5, 1.45) {$\ell_1$};
                  \node (16) at (1.5, 0.15) {$\ell_2$};
    		\node (17) at (-1, 1.25) {};
          
    		\draw [style=T edges] (4) to (1);
    		\draw [style=T edges] (1) to (2);
    		\draw [style=T edges] (7) to (8);
    		\draw [style=T edges] (8) to (9);
      		\draw [style=Fat edge] (7) to (10);
    		\draw [style=witness edge] (9) to (10);
    		\draw [style=witness edge] (3) to (4);
    		\draw [style=T edges] (2) to (7);
    		\draw [style=Fat edge] (10) to (11);
    		\draw [style=Fat edge] (9) to (12);
    		\draw [style=Fat edge] (2) to (3);
    		\draw [style=Fat edge] (3) to (5);
    		\draw [style=Fat edge] (4) to (6);
            \end{tikzpicture}
        \end{tabular}
    \end{center}
    \caption{(a) Ladder $G_s$ adjacent to two lonely nodes and ladder $G_t$ by links $\{\ell,\ell_1,\ell_2\} = \delta(G_s)$. The solution $APX_i$ is shown as bold edges. The red edges are not links. 
    (b)Ladder $G_s$ adjacent to two lonely nodes and ladder $G_t$ by links $\{\ell,\ell_1,\ell_2\} = \delta(G_s)$. The solution $APX_{i+1}$ is shown as bold edges. The addition of $\{e,f\}$ and the removal of $\ell$ is charged to $\{\ell,\ell_1,\ell_2\}$. 
    }
    \label{fig:case1sketch}
\end{figure}
    
    By constructions of $APX_i$ and $APX_{i+1}$, we have 
    \begin{align*}
        APX_{i+1}\backslash (Y_{i+1} \cup E[\cH_{i+1}] )  &= APX_{i}\backslash (Y_{i+1} \cup E[\cH_{i+1}])=  EC_3\backslash (Y_{i+1} \cup E[\cH_{i+1}] ) 
    \end{align*}
     
    To see that $|APX_{i+1}\cap (Y_{i+1} \cup E[\cH_{i+1}])| \leq \frac{4}{3}| EC_3\cap (Y_{i+1} \cup E[\cH_{i+1}])|$, observe that 
    \begin{align*}
        |EC_3\cap (Y_{i+1}\cup E[\cH_{i+1}])| - |EC_3\cap (Y_i\cup E[\cH_i])|\coloneqq \Delta \ge 3
    \end{align*}
    where the inequality follows since $\ell,\ell_1,\ell_2\notin EC_3\cap(Y_i\cup E[\cH_i])$, but are in $EC_3\cap(Y_{i+1}\cup E[\cH_{i+1}])$. Similarly, the following inequality follows since $APX_{i+1}\cap(Y_{i+1}\cup E[\cH_{i+1}])$ contains $\ell_1,\ell_2,e$ and exactly one of $f$ or $\ell$.
    \begin{align*}
        |APX_{i+1}\cap (Y_{i+1}\cup E[\cH_{i+1}])| - |APX_i\cap (Y_i\cup E[\cH_i])| = \Delta +1
    \end{align*}
    So as $i$ increases by one, the increase on the right hand side of the inequality of (2) is at most $\frac{4}{3}$ the increase on the left hand side. Thus, (2) holds.

    Given that we are able to show (1) and (2) hold for the three cases above, we note that for every $i\geq 1$, by (2) we have
    \begin{align*}
        |APX_i| &= |APX_i\cap (Y_i \cup E[\cH_i])| + |APX_i \backslash (Y_i \cup E[\cH_i])| \\
        &\leq \frac{4}{3}| EC_3\cap (Y_i \cup E[\cH_i])| + |EC_3 \backslash (Y_i \cup E[\cH_i])| \leq \frac{4}{3}|EC_3|
    \end{align*}
    Furthermore, in each iteration we cover a new element of $\{G_1,\dots, G_k\}$, so after a polynomial number of iterations we have covered every obstruction.
\end{proof}
We can now combine these ingredients to prove Theorem~\ref{thm:unweighted} (for this simple case). The proof for the general case can be found in Appendix~\ref{sec:4/3approx-proofs}.
\begin{proof}[Proof of Theorem~\ref{thm:unweighted}]
    Given an unweighted $(n-4)$-Node Connectivity Augmentation instance, we apply Theorem~\ref{thm:from_connectivity_to_obstructions} to obtain an unweighted $4$-Obstruction Covering instance, $(G,4,L)$. We will show how to find a solution $APX$ with $|APX| \leq \frac{4}{3}|opt|$. Using Theorem~\ref{thm:ChainDecomposition}, we decompose $G$ into node-disjoint subgraphs $R, G_1,\dots, G_k$. By applying Lemma~\ref{lem:degree123}, we compute a minimum cardinality set of links $EC_3\subseteq L$ that covers the obstructions contained in every subgraph $G_i$ of degree $1$ and $2$, and covers every degree $3$ node. Let $\cH_0\subseteq \{G_1,\dots, G_k\}$ be the ladders whose obstructions are covered by $EC_3$. We set  $APX_1 = EC_3 \cap E[\cH_0]$. We create a new instance $(G'=(V,E\backslash APX_1), 4, L' = L\backslash APX_1)$. 
    Clearly, if we apply Theorem~\ref{thm:ChainDecomposition} to the instance $(G',4,L')$, we get a decomposition that contains only ladders of degree $3$ or $4$.

    We apply Lemma~\ref{lem:simplifiedapproximation} to $(G',4,L')$ with $EC_3\setminus APX_1$ as our partial solution, to find a feasible solution $APX_2$. We let $APX = APX_1 \cup APX_2$. 
    By applying the inequality from the statement of Lemma~\ref{lem:simplifiedapproximation}, we  upper bound $|APX|$ as
    \[|APX_1| + |APX_2| \leq |EC_3\cap E[\cH_0]| + \frac{4}{3}|EC_3\setminus E[\cH_0]| \leq \frac{4}{3}|EC_3| \leq \frac{4}{3}|opt|.\]
\end{proof}
\section{ Hardness of Unweighted Node Connectivity Augmentation}
\label{sec:hardness}
In this section, we will prove that the $d$-Obstruction Covering problem is APX-hard for any $d\ge 4$ with $d = O(n^{c})$ for any fixed constant $c$. 
To prove this, we introduce a variant of the 3-SAT problem, denoted \emph{3-SAT-(2,2)}.

\begin{definition}[3-SAT-(2,2) problem]
    The 3-SAT-(2,2) problem is the special case of 3-SAT where each variable occurs exactly twice as a positive literal and twice as a negative literal in the boolean formula $\I$. The corresponding optimization version of the problem is denoted as MAX-3-SAT-(2,2).
\end{definition}
The following lemma shows that this special case of 3-SAT is as hard as the original problem. Its full proof can be found in Appendix~\ref{sec:3sat22hard}
\begin{lemma}
\label{lem:3sat22hard}
    The 3-SAT-(2,2) problem is NP-hard.
\end{lemma}
%
    %
We will now prove the following preliminary hardness result for our problem.
\begin{theorem}
\label{thm:n-4hard}
    The $4$-Obstruction Covering problem is NP-hard.
\end{theorem}

    In order to prove Theorem~\ref{thm:n-4hard}, we will provide a polynomial-time reduction from 3-SAT-(2,2). Let $\I$ be an instance of 3-SAT-(2,2), we will construct an instance $(G_{\I},4,L_{\I})$ of the $4$-Obstruction Covering problem as follows: 
    \begin{itemize}
        \item For each clause $C$ in $\I$, we add node $C$ to $G_{\I}$. We call these \emph{clause nodes}.
        \item For each variable $x\in I$, we add a subgraph $H_x$ to $G_{\I}$, where $H_x$ is defined by six \emph{variable nodes} $x_1,\dots,x_6$, and edges $\{x_1x_2, x_2x_3, x_3x_4, x_4x_5, x_5x_6,$ $ x_1x_6,x_2x_5\}$. 
        The edges $\{x_1x_2,$ $x_2x_3,$ $x_4x_5,$ $x_5x_6\}$ are included into $L$.
        \item For (possibly equal) clauses $C_1,C_2$ containing the positive literals of $x$, we add edges $\{x_1C_1, x_4C_2\}$ to both $G$ and $L$. For (possibly equal) clauses $C_3,C_4$ containing the negative literals of $x$, we add edges $\{x_3C_3,x_6C_4\}$ to both $G_{\mathcal{I}}$ and $L$.
    \end{itemize}
    
 See Figure~\ref{fig:SAT-obstruction}~(Left) for a depiction of the gadget. By construction, the graph is $3$-regular and no three nodes have two common neighbors, so it is not difficult to show that the defined instance is valid. The proof of the following proposition can be found in Appendix~\ref{sec:GILIfeasible}.

\begin{figure}[!t]
    \centering
    \resizebox{0.9\textwidth}{!}{
        \begin{tikzpicture}
            
            \tikzset{black dot/.style={draw=black, very thick, circle,minimum size=0pt, inner sep=1pt, outer sep=1pt,fill=black}}
            \tikzset{white dot/.style={draw=white, very thick, circle,minimum size=0pt, inner sep=1pt, outer sep=1pt,fill=white}}
            \tikzset{terminal/.style={draw=black,  thick,minimum size=0pt, inner sep=2.5pt, outer sep=1pt}}
            \tikzset{P node/.style={fill={rgb,255: red,20; green,154; blue,0}, draw={rgb,255: red,20; green,154; blue,0}, circle, minimum size=0pt,inner sep=1pt, outer sep=1pt}}
        
            \tikzstyle{witness edge}=[-, draw={rgb,255: red,195; green,0; blue,3}, very thick]
            \tikzstyle{T edges}=[-, very thick]
            \tikzstyle{new witness}=[-, draw={rgb,255: red,195; green,0; blue,3}, dashed, very thick]
            \tikzstyle{connected terminals}=[-, draw=black, dashed, very thick]
            \tikzstyle{P}=[-, draw={rgb,255: red,20; green,154; blue,0}, very thick]
            \node [style=black dot] (0) at (-1.75, 1.75) {};
            \node [style=black dot] (1) at (-1.75, 0.25) {};
            \node [style=black dot] (2) at (0, 0.25) {};
            \node [style=black dot] (3) at (0, 1.75) {};
            \node [style=black dot] (4) at (1.75, 1.75) {};
            \node [style=black dot] (5) at (1.75, 0.25) {};
            \node [style=black dot] (6) at (2.75, -0.5) {};
            \node [style=black dot] (7) at (2.75, 2.5) {};
            \node [style=black dot] (8) at (-3, 2.5) {};
            \node [style=black dot] (9) at (-3, -0.5) {};
            \node (12) at (-2.25, 2.25) {$x$};
            \node (13) at (2, 2.25) {$\bar x$};
            \node (14) at (-2.5, 0) {$\bar x$};
            \node (15) at (2.5, 0) {$x$};
            \node (16) at (-1.45, 2) {$x_1$};
            \node (17) at (0, 2) {$x_2$};
            \node (18) at (1.45, 2) {$x_3$};
            \node (19) at (1.45, 0) {$x_4$};
            \node (20) at (0, 0) {$x_5$};
            \node (21) at (-1.45, 0) {$x_6$};
            \node (22) at (-3.5, 2.75) {$C_1$};
            \node (23) at (3.25, 2.75) {$C_3$};
            \node (24) at (3.25, -0.5) {$C_2$};
            \node (25) at (-3.5, -0.5) {$C_4$};
            \draw [style=T edges] (8) to (0);
            \draw [style=T edges] (1) to (9);
            \draw [style=T edges] (1) to (2);
            \draw [style=T edges] (3) to (4);
            \draw [style=T edges] (5) to (6);
            \draw [style=T edges] (5) to (2);
            \draw [style=T edges] (0) to (3);
            \draw [style=T edges] (7) to (4);
            \draw [style=witness edge] (0) to (1);
            \draw [style=witness edge] (3) to (2);
            \draw [style=witness edge] (4) to (5);
            \node [style=white dot] (blank) at (-2.75, -1) {};
        \end{tikzpicture}    
\hspace{40pt}
        \begin{tikzpicture}
            
            \tikzset{black dot/.style={draw=black, very thick, circle,minimum size=0pt, inner sep=1pt, outer sep=1pt,fill=black}}
            \tikzset{terminal/.style={draw=black,  thick,minimum size=0pt, inner sep=2.5pt, outer sep=1pt}}
            \tikzset{P node/.style={fill={rgb,255: red,20; green,154; blue,0}, draw={rgb,255: red,20; green,154; blue,0}, circle, minimum size=0pt,inner sep=1pt, outer sep=1pt}}
        
            \tikzstyle{witness edge}=[-, draw={rgb,255: red,195; green,0; blue,3}, very thick]
            \tikzstyle{T edges}=[-, very thick]
            \tikzstyle{new witness}=[-, draw={rgb,255: red,195; green,0; blue,3}, dashed, very thick]
            \tikzstyle{connected terminals}=[-, draw=black, dashed, very thick]
            \tikzstyle{P}=[-, draw={rgb,255: red,20; green,154; blue,0}, very thick]
		\node [style=black dot] (0) at (-1.75, 1.75) {};
		\node [style=black dot] (1) at (-1.75, 0.25) {};
		\node [style=black dot] (2) at (0, 0.25) {};
		\node [style=black dot] (3) at (0, 1.75) {};
		\node [style=black dot] (4) at (1.75, 1.75) {};
		\node [style=black dot] (5) at (1.75, 0.25) {};
		\node [style=black dot] (6) at (2.75, -0.5) {};
		\node [style=black dot] (7) at (2.75, 2.5) {};
		\node [style=black dot] (8) at (-3, 2.5) {};
		\node [style=black dot] (9) at (-3, -0.5) {};
		\node (12) at (-2.25, 2.25) {$x$};
		\node (13) at (2, 2.25) {$\bar x$};
		\node (14) at (-2.5, 0) {$\bar x$};
		\node (15) at (2.5, 0) {$x$};
		\node (16) at (-1.45, 2) {$x_1$};
		\node (17) at (0, 2) {$x_2$};
		\node (18) at (1.45, 2) {$x_3$};
		\node (19) at (1.45, 0) {$x_4$};
		\node (20) at (0, 0) {$x_5$};
		\node (21) at (-1.45, 0) {$x_6$};
		\node (22) at (-3.5, 2.75) {$C_1$};
		\node (23) at (3.25, 2.75) {$C_3$};
		\node (24) at (3.25, -0.5) {$C_2$};
		\node (25) at (-3.5, -0.5) {$C_4$};
		\node [style=black dot] (26) at (-1.125, 1) {};
		\node [style=black dot] (27) at (0.625, 1) {};
		\node [style=black dot] (28) at (2.25, 1.25) {};
		\node [style=black dot] (29) at (-1.75, -0.5) {};
		\node [style=black dot] (30) at (-3, 2) {};
		\node [style=black dot] (31) at (2.75, 2) {};
		\node [style=black dot] (32) at (2.75, -1) {};
		\node [style=black dot] (33) at (-3, -1) {};
		\node [style=black dot] (34) at (1.125, 1) {};
		\node [style=black dot] (35) at (-0.625, 1) {};
		\node [style=black dot] (36) at (2, 1.25) {};
		\node [style=black dot] (37) at (-2.75, 2) {};
		\node [style=black dot] (38) at (-2.75, -1) {};
		\node [style=black dot] (39) at (2.5, -1) {};
		\node [style=black dot] (40) at (2.5, 2) {};
		\node [style=black dot] (41) at (-1.5, -0.5) {};
            \node (42) at (-1.25, 0.75) {$y_1$};
            \node (43) at (-0.5, 1.25) {$y_2$};
            \node (44) at (0.5, 0.75) {$z_1$};
            \node (45) at (1.25, 1.25) {$z_2$};
		\draw [style=T edges] (8) to (0);
		\draw [style=witness edge] (0) to (1);
		\draw [style=T edges] (1) to (9);
		\draw [style=T edges] (1) to (2);
		\draw [style=witness edge] (2) to (3);
		\draw [style=T edges] (3) to (4);
		\draw [style=witness edge] (4) to (5);
		\draw [style=T edges] (5) to (6);
		\draw [style=T edges] (5) to (2);
		\draw [style=T edges] (0) to (3);
		\draw [style=T edges] (7) to (4);
		\draw [style=witness edge] (0) to (26);
		\draw [style=witness edge] (26) to (2);
		\draw [style=witness edge] (3) to (27);
		\draw [style=witness edge] (27) to (5);
		\draw [style=witness edge] (4) to (28);
		\draw [style=witness edge] (1) to (29);
		\draw [style=witness edge] (7) to (31);
		\draw [style=witness edge] (6) to (32);
		\draw [style=witness edge] (9) to (33);
		\draw [style=witness edge] (8) to (30);
		\draw [style=witness edge] (3) to (34);
		\draw [style=witness edge] (34) to (5);
		\draw [style=witness edge] (0) to (35);
		\draw [style=witness edge] (35) to (2);
		\draw [style=witness edge] (8) to (37);
		\draw [style=witness edge] (9) to (38);
		\draw [style=witness edge] (4) to (36);
		\draw [style=witness edge] (7) to (40);
		\draw [style=witness edge] (6) to (39);
		\draw [style=witness edge] (41) to (1);
        \end{tikzpicture}    
}
    \caption{Left: Subgraph $H_x$ and clause nodes $C_1,C_2,C_3,C_4$ containing $x$, from the reduction in Theorem~\ref{thm:n-4hard}, for variable $x$ in $\mathcal{I}$. Here the links of $H_x$ are denoted by the black edges. Right: The extension of $\mathcal{I}$ to larger values of $d$ (in this case, $d=6$). It can be seen that the set of links is the same, but the obstructions have a more of nodes.}
    \label{fig:SAT-obstruction}
\end{figure}
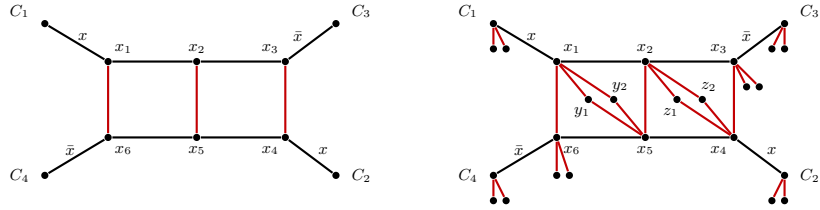
\begin{proposition}
\label{prop:GILIfeasible}
    The graph $G_{\I}$ defined above does not contain neither $K_{1,4}$ nor $K_{2,3}$ as subgraph. Thus, $(G_{\I},4,L_{\I})$ is a valid instance of the $4$-Obstruction Covering problem.
\end{proposition}

    


            
        

    In the following two lemmas we show how to translate a feasible assignment of values for instance $\I$ into an optimal solution of $(G_{\I},4,L_{\I})$ and back.  
    
\begin{lemma}
\label{lem:n-4to3sat22}
    Let $k$ be the number of variables in the 3-SAT-(2,2) instance $\I$. If $\I$ is satisfiable, then there is a feasible solution to the $4$-Obstruction Covering instance $(G_{\I},4,L_{\I})$ of size $4k$.
\end{lemma}
\begin{proof}
    Consider a satisfying assignment for instance $\I$. For each variable $x$ in $\I$, if $x$ is set to true (resp. false), then for the (possibly equal) clauses $C_1$ and $C_2$ containing the positive (resp. negative) literals of $x$, we take the links $x_1C_1$ and $x_4C_2$ (resp. $x_3C_3,x_6C_4$), as well as the links $\{x_2x_3,x_5x_6\}$ (resp. $\{x_1x_2, x_4x_5\}$). This way, since the assignment satisfies $\I$, all the obstructions $K_{1,3}$ defined by clause nodes are covered. The $K_{1,3}$ obstructions defined by variable nodes and the $K_{2,2}$ obstructions of each $H_x$ are also covered by the links $\{x_2x_3,x_5x_6\}$ or $\{x_1x_2, x_4x_5\}$ depending on the case. Therefore, we have a feasible solution for the $4$-Obstruction Covering instance $(G_{\I},4,L_{\I})$ of size $4k$. 
\end{proof}
\begin{lemma}\label{lem:3sat22ton-4}
    Let $k$ be the number of variables in the 3-SAT-(2,2) instance $\I$. If there is a feasible solution to the $4$-Obstruction Covering instance $(G_{\I},4,L_{\I})$ of size $4k$, then $\I$ is satisfiable.
\end{lemma}
\begin{proof}
    Note, for a given variable $x$ in $\I$, any links taken to cover obstructions in subgraph $H_x$ do not cover obstructions from subgraph $H_y$ for variable $y\neq x$. Hence, the size of any feasible solution $F \subseteq L_{\I}$ to $(G_{\I},L_{\I})$ is equal to $\sum_{x\in \I} : \{\ell \in F | \ell \text{ incident to } v\in V(H_x) \}|$. 
    Observe that every feasible solution must contain at least four links incident to a nodes $V(H_x)$ as we must have at least one link incident to every vertex in $\{x_1,x_2,\ldots,x_6\}$. Therefore in a feasible solution of size $4k$, for each variable $x\in \I$ we have exactly four links with an endpoint in $V(H_x)$. Furthermore for each variable $x$, such a solution has exactly one link incident on each $x_i$ for every $i\in \{1,3,4,6\}$.

    We will prove we can assume the feasible solution of size $4k$ to $(G_{\I},L_{\I})$ has the following structure: for each subgraph $H_x$, either we have the links between $x_1$ and $x_4$ and their corresponding adjacent clause nodes plus $x_2x_3$ and $x_5x_6$, or the links between $x_3$ and $x_6$ and their corresponding adjacent clause nodes plus $x_1x_2$ and $x_4x_5$. 
    
    Suppose that for some variable $x\in \I$, we cover the obstructions of $H_x$ with a different set of links.

    \textbf{Case 1, we choose at least three links from $H_x$:} Let us assume without loss of generality that these links are $\{x_1x_2,x_2x_3,x_4x_5\}$. Since the solution is feasible, it must contain either the link between $x_6$ and its adjacent clause node or the link $x_5x_6$. In the first case, we can replace $x_2x_3$ with the link joining $x_3$ and its adjacent clause node. In the second case, we can replace both $x_2x_3$ and $x_5x_6$ with the corresponding links joining them to their adjacent clause nodes, and in both cases every obstruction in $H_x$ is still covered.

     \textbf{Case 2, case 1 does not apply:} We need to cover both of the $K_{2,2}$ obstructions plus $x_2$ and $x_5$. In order to do it with two links only, one of them must be incident to $x_2$ and the other to $x_5$. Without loss of generality we assume our solution contains $x_1x_2$ and $x_4x_5$ and it does not contain $x_2x_3$ and $x_5x_6$. Therefore as we need to cover $x_3$ and $x_6$ we must have the links from $x_3$ and $x_6$ to the clause nodes in our solution. Thus our solution forms the desired configuration.
     
    
     Notice that, since the solution for $(G_{\I},4,L_{\I})$ is feasible, every clause node is covered. It is possible to find a satisfying assignment to the variables of $\I$ in the following way: For each variable $x$, if the links incident to $x_1$ and $x_4$ and their adjacent clause nodes are taken, then we set $x$ to be true. If the links $x_3$ and $x_6$ and their adjacent clause nodes are taken instead, then we set $x$ to be false.

    To see that this assignment satisfies $\I$, suppose for the sake of contradiction that some clause $C = (a\lor b\lor c)$ is not satisfied. Then, none of the links between $H_a,H_b,H_c$ and $C$ are taken, and thus $C$ is uncovered. This means that the solution is not covering every obstruction in $(G_{\I},L_{\I})$, which is a contradiction.

\end{proof}
We can now proceed with the proof of Theorem~\ref{thm:n-4hard}.

\begin{proof}[Proof of Theorem~\ref{thm:n-4hard}.] 
    Consider an instance $\I$ of 3-SAT-(2,2) defined by $k$ variables. We construct the $(n-4)$-Obstruction Covering instance $(G_{\I},4,L_{\I})$ as described before. Thanks to Lemmas~\ref{lem:n-4to3sat22} and \ref{lem:3sat22ton-4}, $\I$ is satisfiable if and only if $(G_{\I},4,L_{\I})$ admits a feasible solution of size $4k$. Since 3-SAT-(2,2) is NP-hard (Lemma~\ref{lem:3sat22hard}), the $4$-Obstruction Covering problem is NP-hard as well. 
\end{proof}

It is possible to extend the previous hardness results to further values of $d$, as the following theorem states (see Figure~\ref{fig:SAT-obstruction}~(Right)). The proof of this theorem can be found in Appendix~\ref{sec:n-dobscovhard}

\begin{theorem}\label{thm:n-dobscovhard}
    For any given $d\ge 4$, $d\in O(n^{c})$ for any fixed constant $c>0$, the $d$-Obstruction Covering problem is NP-hard.
\end{theorem}

It is possible to improve the previous results to prove that the $d$-Obstruction Covering problem is indeed APX-hard. Our results follow from the fact that a closely related problem to MAX-3-SAT-(2,2), known as MAX-3-SAT-4, is APX-hard~\cite{berman2003approximation}. A detailed proof can be found in Appendix~\ref{sec:obscovapxhard}.

\begin{theorem}\label{cor:n-dobscovapxhard} 
    For any given $d\ge 4$, $d \in O(n^{c})$ for any fixed constant $c>0$, the $d$-Obstruction Covering problem is APX-hard.
\end{theorem}

\bibliographystyle{plainurl}
\bibliography{refsNOFIRSTNAMES}

\newpage
\begin{appendix}
    \section{ Previous hardness results for Vertex-Connectivity Augmentation}
\label{sec:appendix_prevhard}

Kortsarz et al.~\cite{KortsarzKL2004} prove the following hardness result for $1$-node connectivity augmentation.

\begin{theorem}[\cite{KortsarzKL2004}] There exists $\varepsilon_0>0$ fixed, such that it is NP-hard to approximate $1$-Node Connectivity Augmentation within a factor better than $1+\varepsilon_0$. 
\end{theorem}

This result is proved via a reduction from a bounded degree version of 3D Matching. In order to obtain analogous hardness results for $k$-Node Connectivity Augmentation, it is argued that we can iteratively add $k-1$ nodes to the construction which are connected to all the other nodes in the graph (this way, the graph becomes $k$ connected and the cuts that must be covered are exactly the same). Thus, in general, if we add $t$ dummy nodes to the instance, which has $n$ nodes, we get a graph with $n'=n+t$ nodes which is $(t+1)$-node connected.

In order for this procedure to be a polynomial-time reduction, it is required that the value of $t$ is polynomially bounded with respect to $n$, let us say at most $n^c$ for some constant $c>0$, which implies that the new instance we obtain is a $(n'-n'^{1/c})$-connected graph and we wish to augment its connectivity by one. 

In particular, it is not possible for this procedure to reach a $(n'-d)$-connected instance where $d$ is constant or logarithmic, and to have polynomial running time at the same time. This is also consistent with the fact that $(n-3)$-node connectivity augmentation is polynomial-time solvable~\cite{BERCZI2012565}.
    \section{ Dynamic Program for independent ladders and hexagons}
\label{sec:DP}

In this section we prove Theorem~\ref{alg:DP} which we restate:
\DP*

Assume $u_1,\ldots,u_{r+1},v_1,\ldots,v_{r+1}$ is the given consistent labelling of $C$ given by Theorem~\ref{thm:ChainDecomposition}.
If $C$ has at most $10$ nodes, then we can enumerate over the possible edge covers to find the minimum cost covering $F$. Notice that a hexagon has 7 nodes, thus, we can assume that $C$ is a ladder of length at least $6$.


We denote by $D$ the possible edges in $C$ with both endpoints in $\{v_1, u_1,$ $v_{r+1},u_{r+1}\}$. We will compute a different solution for every possible $2^{|D|}$ combination of edges from $D$ in $F$. If we have $e = xy \in D\cap F$, then we set $h(x)=h(y)=0$, as these endpoints are covered by $e\in F$. Furthermore, if $v_1u_1 \in D\cap F$, then we set $w(v_1u_1) = 0$ to ensure we have $v_1u_1$ in the solution and we avoid double counting. The case where $v_{r+1}u_{r+1} \in F\cap D$, is handled similarly. Lastly, if the edges of $D$ form a square $B$, then we require that $|F\cap D| \geq 1$. Observe that since $r\ge5$, other than the squares $S_i$, the only possible square of $C$ is $B$.

Before we define our dynamic program, in order to simplify notation, for each $S_i$, $i=1,\dots,k$, we define $t_i \coloneqq v_iv_{i+1}, b_i \coloneqq u_iu_{i+1}, L_i \coloneqq v_iu_i,$ and, $L_{i+1} \coloneqq v_{i+1}u_{i+1}$. We are now ready to define our dynamic program.


We compute from $S_1$ to $S_i$ the minimum number of links needed to cover obstructions $S_1, \dots, S_{i}$ and satisfy the requirements $h(v_j)$ and $h(u_j)$ for $j<i$, given whether $t_{i},b_i,L_{i+1}$ are in $F$ or not. We let  $\X_{L_i}, \X_{t_i}, \X_{b_i} \in \{0,1\}$ be the indicator variable for whether our solution includes $L_i,t_i,$ and $b_i$ respectively or not. We store our solution in the table $A[i, \X_{L_{i+1}}, \X_{t_{i}}, \X_{b_{i}}]$. 

For $i=1$, $A[1, \X_{L_2}, \X_{t_1}, \X_{b_1}]$ is equal to 
\begin{align*}
    \min_{ \X_{L_{1}}, \X_{t_1}, \X_{b_1}\in \{0,1\}} 
    \left\{w(L_1)\X_{L_1} + w(L_2)\X_{L_{2}} + w(t_1)\X_{t_1} + w(b_1)\X_{b_1}\right\}
\end{align*}
Subject to the requirements that (1) $\X_{t_{1}} \lor \X_{L_{2}}=1$ if $h(v_{1})=1$ to ensure that the coverage requirement of $v_{1}$ is met, (2) $\X_{b_1}\lor \X_{L_{2}} =1$ if  $h(u_{1})=1$ to ensure that the coverage requirement of $u_{1}$ is met, and (3) $\X_{t_{1}}\lor \X_{b_{1}} \lor \X_{L_{2}} = 1 $ or $u_1v_1\in F\cap D$ to ensure that $S_{1}$ is covered.
Since we are considering all the possible cases by enumeration it should be clear that $A[1, \X_{L_2}, \X_{t_1}, \X_{b_1}]$ is computed correctly.

For $i\in \{1,\dots, r-1\}$ we set
 $A[i+1, \X_{L_{i+2}}, \X_{t_{i+1}}, \X_{b_{i+1}}]$ to 
\begin{align*}
    \min_{\X_{L_{i+1}}, \X_{t_i}, \X_{b_i}\in \{0,1\}}
    \{&w(L_{i+1})\X_{L_{i+1}} + w(t_{i+1})\X_{t_{i+1}} + w(b_{i+1})\X_{b_{i+1}} \\ +& A[i, \X_{L_{i+1}}, \X_{t_i}, \X_{b_i}] \}
\end{align*}
Subject to the following requirements: (1) $ \X_{t_i}\lor \X_{t_{i+1}} \lor \X_{L_{i+1}}= 1$ if  $h(v_{i+1})=1$ to ensure that the coverage requirement of $v_{i+1}$ is met, (2) similarly, $\X_{b_i}\lor \X_{b_{i+1}} \lor \X_{L_{i+1}} =1$ if  $h(u_{i+1})=1$ to ensure that the coverage requirement of $u_{i+1}$ is met, and (3) $\X_{t_{i+1}}\lor \X_{b_{i+1}} \lor \X_{L_{i+1}}\lor \X_{L_{i+2}} = 1 $ to ensure that $S_{i+1}$ is covered.

Now by induction we can show that $A[i, \X_{L_{i+1}}, \X_{t_{i}}, \X_{b_{i}}]$ is computed correctly. We know that this is true for $i=1$. Assume we have computed all the cells $A[i-1, \X_{L_{i+2}}, \X_{t_{i+1}}, \X_{b_{i+1}}]$ for some $i$. For each configuration of $\X_{L_{i+1}}, \X_{t_{i}}, \X_{b_{i}} \in \{0,1\}$, we can apply the inductive hypothesis to $A[i, \X_{L_{i+1}}, \X_{t_i}, \X_{b_i}]$, to find the minimum cost number of links needed to cover obstructions $S_1,\dots, S_i$, and satisfy requirements $h(v_j)$ and $h(u_j)$ for $j<i$. We take the minimum of the configurations of $\X_{L_{i+1}}, \X_{t_{i}}, \X_{b_{i}} \in \{0,1\}$ and the claim holds.

To find the minimum solution, we find 
$$\min_{\X_{L_{r+1}},\X_{t_{r}}, \X_{b_r} \in \{0,1\}} A[r,  \X_{L_{r+1}}, \X_{t_{r}}, \X_{b_r}],$$ subject to the requirements
\begin{itemize}
    \item  $\X_{t_r} \lor \X_{L_{r+1}} =1$ if $h(v_{r+1})=1$,
    \item $\X_{b_r} \lor \X_{L_{r+1}} = 1$ if $h(u_{r+1})=1$, and
    \item $x_{L_{r+1}}=0$ if $u_{r+1}v_{r+1}\notin F\cap D$.
\end{itemize}

    \section{ Subset-Edge Cover}
\label{sec:SEdgeCover}
\SEdgeCover*
\begin{proof}
    If $N=V(G)$, then we can solve the problem by finding the minimum edge cover of $G$, which is known to be solvable in polynomial time.
    Now assume $N\subsetneq V(G)$. In this case we construct a graph $G'$, that contains $G$ and has an additional vertex $v$. Now for every $u\in V(G)\setminus N$, we add an edge of weight zero from $v$ to $u$ in $G'$. We compute the minimum weight edge cover $E'$ of $G'$. Assume $E_v$ is the set of all edges incident to $v$ in $G'$. We argue that $E'':=E\setminus E_v$ is the minimum weight $N$-edge cover of $G$. For this purpose it is sufficient to observe that if $E_1$ is an $N$-edge cover for $G$ then $E_1\cup E_v$ is an edge cover of $G'$. Similarly if $E_2$ is an edge cover for $G'$ then $E_2\setminus E_v$ is an edge cover of $G$. Therefore the cost of the minimum $N$-edge cover of $G$ is indeed equal to the cost of the minimum edge cover of $G$. Therefore our solution, $E'':=E\setminus E_v$ is a minimum $N$-edge cover.
\end{proof}

    \section{ Proof of Theorem~\ref{thm:from_connectivity_to_obstructions}}
\label{sec:from_connectivity_to_obstructions}
In this section we prove Theorem~\ref{thm:from_connectivity_to_obstructions}. We remark that~\cite{BERCZI2012565} shows this for $d=4$. Here we generalize this for other values of $d$ (see~\cite{BercziV10} for a similar discussion). To prove the Theorem we first prove the following lemma.
\begin{lemma}\label{lem:CutsInTheComplement}
     Let $G=(V,E)$ be a connected graph. $G$ has a vertex cut $V'$ of size $r $ if and only if $\Bar{G}$ contains a subgraph $K_{i,|V|-r-i}$ for some $1\le i < |V|-|V|'$.

\end{lemma}

\begin{proof}
    Assume $V'$ is a cut of size $r$ in $G$, then $G[V\setminus V']$ has at least two connected components. Let $C$ be a connected component of $G[V\setminus V']$. Then in $\Bar{G}$, the induced subgraph on $V\setminus V'$ contains a $K_{|C_1|,V-|V'|-|C_1|}$.
    Now assume $\Bar{G}$ contains a subgraph $K_{i,|V|-r-i}$ for some $1\le i < |V|-|V|'$. Let $(X,Y)$ be the bipartition of this subgraph. We argue that $V':=V\setminus (X\cup Y)$ is a cut of size $r$. Observe that by deleting $V'$, there will be no edges between $X$ and $Y$. Furthermore $|V'|=|V|-(|X|+|Y|)=|V|-(|V|-r)=r$. Thus $V'$ is clearly a cut of size $r$.
\end{proof}


\begin{proof}[Proof of Theorem~\ref{thm:from_connectivity_to_obstructions}]
Given an instance $I$ of $(n-d)$-node-connectivity Augmentation Problem on graph $G=(V,E)$, with link set $L$ and cost $c:E\rightarrow \mathbb{R}$ for the edges we obtain the instance of $I':=(\Bar{G},d,L,c)$ of $d$-Obstruction Covering problem in which $\Bar{G}$ is the complement of $G$. Observe that, using Lemma~\ref{lem:CutsInTheComplement}, $\Bar{G}$ contains no forbidden subgraphs.
Hence, again by Lemma~\ref{lem:CutsInTheComplement}, $G\cup L'$ is $(n-d+1)$-vertex-connected if and only if $L'$ has a link in every obstruction of $\bar G$.
This implies that for every subset $L'$ of $L$,  $L'$ is a feasible solution for $I$ if and only if $L'$ is a feasible solution for $I'$.  


\end{proof}

\section{Proof of Theorem~\ref{thm:ChainDecomposition}}
\label{sec:ChainDecomposition}
In order to prove Theorem~\ref{thm:ChainDecomposition}, we need the following technical lemmas. 
\begin{lemma}\label{lem:HexagonExtension}
    Let $H$ and $S$ be hexagon and square subgraphs of $G$ respectively, such that $V(S)\cap V(H)\neq\emptyset$, and $V(S) \not \subseteq V(H)$. Then, in polynomial time, we can find a ladder $C$ in $G$ of length three, and a consistent labelling of $C$, such that $V(C) = V(H) \cup V(S)$.
\end{lemma}
\begin{proof}
    Denote the nodes of $S$ by $uwzv$, and without loss of generality assume $v \in V(S)\cap V(H)$, and $w\in V(S)\backslash V(H)$. 
    If $S\cap H = \{v\}$, then $v$ is a corner of $H$ that is adjacent to two nodes in $H$ and two nodes in $S\backslash H$, so $v$ has degree $4$ which is a contradiction.
    If $S\cap H = \{u,v\}$, then they cannot be adjacent or else they will again have degree $4$.
    So $S\cap H = \{u,v,z\}$ where $u$ and $z$ are corners of $H$. Which implies that $v$ is not a corner in $H$. 
    It is not hard to see that $S$ combined with two of the squares in $H$ gives a ladder $C$ of length three such that $V(C) = V(H)\cup V(S).$ Note that as $|C|$ is constant in size, then by enumeration we can compute a consistent labelling of $C$.  
\end{proof}

\begin{lemma}\label{lem:ChainExtension}
    Let $C$ be a ladder in $G$ with a consistent labelling. Let $S$ be a square in $G$, such that $V(S)\cap V(C) \neq\emptyset$, and $V(S) \not \subseteq V(C)$. Then, in polynomial time, we can find a subgraph $C'\subseteq G$, such that $V(C)\subsetneq V(C')$, where $C'$ is either a ladder with a consistent labelling or a hexagon. 
\end{lemma}
\begin{proof}
    If the length of $C$ is 1  then $C\cup S$ is a ladder of length 2, as $G$ does not contain a  $K_{2,3}$ forbidden subgraph nor a  $K_{1,4}$ forbidden subgraph. In this case, since $|C|$ is constant in size, we can compute a consistent labelling.

    So assume the length of $C$ is at least $ 2$, and hence $r \ge 3$.
    Let $\{u_1, \dots, u_r,$  $v_1, \dots, v_r\}$ be the consistent labelling of $C$.
    Since $V(S) \not \subseteq V(C)$ and the degree of nodes in $G$ is at most $3$, then $S$ contains edges $xy$ and $x'y'$ such that $x$ and $x'$ are distinct corners of $S$ and $y$ and $y'$ are  (not necessarily distinct) vertices in $V(S)\setminus V(C)$. Observe that $x$ and $x'$ must have precisely two neighbours from $V(C)$ as they are degree $3$ in $G$, and are thus corners of $C$.

    Without loss of generality, we assume $x=u_r$. We distinguish the following cases:
    1) If $x'=v_r$, then $v_ru_ryy'$ must be a square in $G$. Therefore $C':=S\cup C$ is a ladder of size $r$. We can label $u_{r+1} \coloneqq y$ and $v_{r+1}\coloneqq y'$.
    2) If $x'=v_1$, then in the square $S$ there must be a path of length at most two in $C$ from $u_r$ to $v_1$ in $C$. However this contradicts the fact that a such path does not exist as $r\ge3$.
    3) If $x'=u_1$, then in the square $S$ there must be a path $P$ of length at most two in $C$ from $u_r$ to $v_1$ in $C$. As $r\ge 3$, in order for such a path to exist we must have $r=3$, and $P=u_1u_2u_3$. Therefore $S=u_1u_2u_3y$, and we have $y=y'$. In this case $S\cup C$ is a Hexagon.  
\end{proof}
With Lemmas~\ref{lem:HexagonExtension} and \ref{lem:ChainExtension}, we can prove Theorem~\ref{thm:ChainDecomposition}.
\begin{proof}[Proof of Theorem~\ref{thm:ChainDecomposition}]
    We find the decomposition by performing a sequence of steps.
    At the $i$-th step, assume we have subgraphs  $R_i, G_1,G_2,\ldots,G_{i-1}$ of $G$ such that:
    \begin{enumerate}
        \item $R_i, G_1,G_2,\ldots,G_{i-1}$ are pairwise node-disjoint
        \item Each $G_j$ is either a ladder with a consistent labelling or a hexagon.
        \item For every square $S$ of $G$ and every $1\le j< i$, if $V(S)\cap V(G_j) \neq \emptyset$ then $V(S) \subseteq V(G_j)$.
        \item $V(R_i) = V \backslash (V(G_1) \cup V(G_2) \ldots \cup V(G_{i-1}))$.
    \end{enumerate}
    Initially, we set $i=1$, and note that $R_1=G$. 
    At a given step $i$,
    if $V(R_i)$ contains only lonely nodes, then $R_i, G_1,G_2,\ldots,G_{i-1}$ are the desired subgraphs. 
    Otherwise we can find a square $S$ such that $V(S)$ and $ V(G_1)\cup V(G_2) \ldots \cup V(G_{i-1})$ are disjoint, by property (3). $S$ itself is a ladder of length one, and it is easy to find a consistent labelling for $S$. 
    Now, we extend $S$ to a maximal ladder or hexagon in $R_i$ by repeatedly applying Lemmas~\ref{lem:HexagonExtension} and \ref{lem:ChainExtension}, as we can enumerate in polynomial time all squares and hexagons in $R_i$, and check if they intersect with the nodes of $S$. In this way, we can find in polynomial time a ladder or a hexagon $G_i$, such that $V(S)\subseteq V(G_i) \subseteq V(R_i)$ and for every square $S'$ of $G$, if $V(S')\cap V(G_i) \neq \emptyset$ then $V(S') \subseteq V(G_i)$. Note that if $G_i$ is a ladder, then Lemmas~\ref{lem:HexagonExtension} and \ref{lem:ChainExtension} also yield a consistent labelling. Furthermore, by construction we keep that $G_1,G_2,...,G_{i}$ are node-disjoint. Now we increase $i$ by one and go to the next step. As the size of $V(R_i)$ is decreasing at each step, the process terminates in polynomial time. At termination, we reach the desired subgraphs.
\end{proof}
    \section{ Missing Proofs of Section~\ref{sec:3/2approx}}
\label{sec:3/2approx-proofs}
\subsection{ Proof of Corollary~\ref{cor:4regularFactorization}}
\label{sec:4regularFactorization}
We first recall Petersen's Theorem
\begin{theorem}[Petersen's Theorem \cite{Pet1981}]\label{thm:Petersen}
    Let $G$ be a $2k$-regular graph, for some integer $k$. Then, $G$ can be decomposed into $k$ $2$-regular spanning subgraphs.
\end{theorem}
We are now ready to prove Corollary~\ref{cor:4regularFactorization}.
\begin{proof}[proof of Corollary~\ref{cor:4regularFactorization}]
    Using Petersen's Theorem $G$ has a $2$-regular spanning subgraph. We can compute one such subgraph $G_1$ in polynomial time by computing a maximum $2$-matching, which is well known to be polynomial time. Now $G_1$ and $G_2:=G\setminus E(G_1)$ are the desired subgraphs.
\end{proof}

\subsection{ Making \texorpdfstring{$G'$}{} 4-regular}
\begin{lemma}
\label{lem:deg4}
    Let $G'=(V',E')$ be the graph constructed in Section~\ref{sec:3/2approx}. If $G'$ is not $4$-regular, then we can add a polynomial number of dummy nodes to $V'$ and zero cost dummy edges to $E'$ such that $G'$ becomes $4$-regular.
\end{lemma}
\label{sec:deg4}
\begin{proof}
    For every $g_i \in V'$ with degree less than $4$ in $G'$, we add a dummy node $\bar g_i$ to $V'$ and  (possibly parallel) zero cost edges $g_i\bar g_i$ to $E'$ so that $g_i$ has degree $4$. Observe that the dummy nodes in $V'\backslash V(g_1 \cup \dots \cup g_k) $ have degree at most $4$. For every dummy node $v\in V'$ with degree $1$ or $2$, we add a self loop of cost zero to $E'$, increasing the degree of $v$ by 2, so that $v$ has degree $3$ or $4$.

    As the sum of degrees of nodes in $G'$ is even, and all nodes in $V'$ have degree $3$ or $4$, we can see that there must be an even number of degree $3$ dummy nodes. Therefore, we can pair up these dummy nodes in an arbitrary manner and add zero cost edges to $E'$ between these pairs so that every node in $G'$ has degree $4$.
\end{proof}

\subsection{ Proof of Lemma~\ref{lem:reducedegrees}}
\label{sec:reducedegrees}
    Let $G_1',\dots, G_k'$ be the subgraphs of $G''$ found by Theorem~\ref{thm:ChainDecomposition}. 
    \begin{enumerate}
        \item We consider cases defined by the edges of $H$  that are incident to $G_i'$:
        \begin{itemize}
            \item $H$ has a loop on $G_i$. Then there is a link $f \in E[G_i']$ that is the unique link in $F_1$ with an endpoint in $V(G_i')$. For each endpoint of $f$ that is a corner, by the construction of $G''$ we create dummy nodes of these corners and change the incident links in $\delta(G_i')$ with links incident to these dummy nodes instead. Therefore, $G_i'$ will have degree $1$ or $2$ in $G''$.
            \item $H$ has two distinct (possibly parallel) edges incident to $G_i'$, denoted $f_1$ and $f_2$. Recall that we can observe that each corner is satisfied by exactly one link $\ell_i$. Note that for every $G_i'\in \{G_1, \dots,G_k\}$ that has degree $3$ or $4$, we can see that that in $G''$, $G_i'$ is incident to $4-|\delta(G_i')|\in \{0,1\}$ dummy edges using this observation.
            Therefore, if the degree of $G_i'$ in $G$ is 3, then $H$ must contain at least one edge whose label is incident to $G_i'$. Similarly, if the degree of $G_i'$ in $G$ is $4$, then both edges incident to $G_i'$ in $H$ have labels that are incident to $G_i'$ in $G$. In either case, we reduce the degree of $G_i'$ to $2$.
        \end{itemize}
        Therefore, when we create $G''$ by either removing edges in $\delta(G_i')\cap F_1$, or creating dummy nodes so that an edge in $\delta(G_i')$ is not incident to that dummy node, we reduce the degree of $G_i'$ to $1$ or $2$.

        \item Let $opt''$ denote the optimal solution to $(G'',4,L'',c'')$. The set of obstructions of $G''$ is a (not necessarily strict) subset of the obstructions of $G$, so edges of $opt$ (or corresponding edges created by replacing nodes with dummy nodes) define a feasible solution to $(G'',4, L'',c'')$. Therefore, $c(opt'') \leq c(opt)$.
    
        \item Given $F$, we can find a partial solution $F_2$ to $(G,4,L,c)$ by taking the edges of $F$ that are in both $E''$ and $E$, and the edges of $F\cap E''$ that have a label in $E$. Clearly we have $c(F_2) = c''(F)$.
        The only obstructions that are uncovered by $F_2$ are a subset of $N$, which are covered by $F_1$. Therefore, we have a solution for $(G,L,c)$ of cost $c(F_2) + c(F_1) = c''(F) + c(F_1).$ 
    \end{enumerate}

\subsection{ Proof of Lemma~\ref{lem:degree12edgecover}}
\label{sec:degree12edgecover}
To prove Lemma~\ref{lem:degree12edgecover}, we prove a slightly more general statement. We are given a $4$-Obstruction Covering instance $(G,4,L,c)$ with subgraphs $G_1,\dots, G_k$ found by Theorem~\ref{thm:ChainDecomposition}, and show that we can find in polynomial time a partial solution $F$ that covers every degree $1$ or $2$ $G_i\in \{G_1,\dots, G_k\}$ and every degree $3$ node for minimum cost $c(F)$.  

To see this, we will first show that there is a polynomial time reduction to an instance of the $N$-Edge Cover problem $G'=(V',E')$ with edge costs $c'$ and optimal solution $F'$, such that $c'(F')=c(F)$. 

First, we extend the definition of $c$ from $L$ to $E$ by setting $c(e)=\infty$ for all $e\in E\backslash L$. 

To construct the $N$-Edge Cover instance $G'=(V',E')$ with edge costs $c'$, we will replace each degree 1 or 2 $G_i\in \{G_1,\dots, G_k\}$ with a certain subgraph with certain edge costs, which we will call a degree $1$ gadget or degree $2$ gadget. We will add all degree $3$ nodes to $N$, and a subset of the nodes of the gadgets to $N$ and show that the optimal $N$-edge cover $F'$ has cost $c'(F')=c(F)$.

Initially, $G'$ is equal to the subgraph induced on $G$ by the nodes that are not in any degree 1 or 2 $G_i\in \{G_1,\dots, G_k\}$. We then consider degree 1 $G_i\in \{G_1,\dots, G_k\}$ with corner $v$. We apply Theorem~\ref{alg:DP} twice to the subgraph induced on $V(G_i)$, subject to the following requirements: 1) $v$ is covered by a link in $E[G_i]$, and 2) $v$ is not covered by a link in $E[G_i]$.

\begin{itemize}
    \item For the case where we require that $v$ is covered by a link in $E[G_i]$, we define requirement function $h:V(G_i) \rightarrow \{0,1\}$ by $h(v)=1$, and for $u\in V(G_i)\backslash\{v\}$, $h(u)=1$ if $u$ has degree $3$, and $h(u)=0$ otherwise. We then apply Theorem~\ref{alg:DP} to the subgraph induced on $V(G_i)$, with edge weights $w(e)=c(e)$ for $e\in E[G_i]$ that do not have $v$ as an endpoint and $w(e)=\infty$ for $e\in E[G_i]$ that have $v$ as an endpoint, and requirement function $h$. We denote the cost of this by $c_1^i$. 

    \item For the case where we require that $v$ is not covered by a link in $E(G_i)$, we define requirement function $h':V(G_i) \rightarrow \{0,1\}$ by $h'(v)=0$, and for $u\in V(G_i)\backslash\{v\}$, $h'(u)=1$ if $u$ has degree $3$, and $h'(u)=0$ otherwise. We then apply Theorem~\ref{alg:DP} to the subgraph induced on $V(G_i)$, with edge weights $w'(e)=c(e)$ for $e\in E[G_i]$, and requirement function $h'$. We denote the cost of this by $c_0^i$.
\end{itemize}

We then add the nodes $v$, $v_1$, and $v_0$ to $V'$, and add the edges $v_0v_1$ and $v_1v$ to $E'$, as well as the edge in $\delta(G_i)$ to $E'$ (as this is a single edge incident to $v$). We give these edges costs $c'(v)=c(v)$, $c'(v_0v_1) = c_0^i$ and $c'(v_1v) = c_1^i$(see Figure~\ref{fig:degree1}~(a) and (b)).We add $v$ and $v_1$ to $N$.  We call this a degree $1$ gadget, with costs $c_0^i$ and $c_1^i$. \input{./Figures/degree1} We repeat this process of finding degree 1 gadgets for each degree 1 $G_i\in \{G_1,\dots, G_k\}$.

Now consider degree 2 $G_i\in \{G_1,\dots, G_k\}$, with corners $v_1$ and $v_2$. We apply Theorem~\ref{alg:DP} four times to the subgraph induced on $V(G_i)$, each time subject to one of the following requirements: 1) both $v_1$ and $v_2$ are covered by links in $E[G_i]$; 2) $v_1$ is covered by links in $E[G_i]$ but $v_2$ is not; 3) $v_2$ is covered by links in $E[G_i]$ but $v_1$ is not, and; 4) neither $v_1$ nor $v_2$ are covered by links in $E[G_i]$.

For each $v\in V(G_i)$, we let $h(v) = 1$ if $v$ is degree $3$ and $0$ otherwise. We apply Theorem~\ref{alg:DP} to the subgraph induced on $V(G_i)$ in the following ways: 
\begin{enumerate}
    \item For the case that neither $v_1$ nor $v_2$ are covered by links in $E[G_i]$, we define requirement function $h_0 : V(G_i) \rightarrow \{0,1\}$ by $h_{0}(v_1)= h_0(v_2)=0$, and for $u \in V(G_i)\backslash\{v_1,v_2\}$, $h_{0}(u) = 1$ if $u$ is degree $3$, and $h_{0}(u)=0$ otherwise.  We define edge costs $w_0 : E[G_i] \rightarrow \mathbb{R}_{\geq 0}$, with  $w_{0}(e)=c(e)$ for all $e\in E[G_i]$ that are not incident to $v_1$ or $v_2$, and $w_0(e)=\infty$ for $e\in E[G_i]$ that are incident to $v_1$ and $v_2$. We apply Theorem~\ref{alg:DP} to the subgraph induced on $V(G_i)$ with edge costs $w_{0}$, and requirement function $h_{0}$, and denote the resulting cost by $c_{0}^i$.
    
    \item For the case that both $v_1$ and $v_2$ are covered by links in $E[G_i]$, we define requirement function $h_{12}: V(G_i) \rightarrow \{0, 1\}$ by $h_{12}(v_1)=h_{12}(v_2)=1$, and for $u \in V(G_i)\backslash\{v_1,v_2\}$, $h_{12}(u) = 1$ if $u$ is degree $3$, and $h_{12}(u)=0$ otherwise. We define edge costs $w_{12}: E[G_i] \rightarrow \mathbb{R}_{\geq 0}$, with  $w_{12}(e)=c(e)$ for all $e\in E[G_i]$. We apply Theorem~\ref{alg:DP} to the subgraph induced on $V(G_i)$ with edge costs $w_{12}$, and requirement function $h_{12}$, and denote the resulting cost by $c_{12}^i$.
    
    \item For the case that $v_1$ is covered by links in $E[G_i]$ but $v_2$ is not, we define requirement function $h_1 : V(G_i) \rightarrow \{0,1\}$ by $h_{1}(v_1)=1$, and $h_{1}(v_2)=0$, and for $u \in V(G_i)\backslash\{v_1,v_2\}$, $h_{1}(u) = 1$ if $u$ is degree $3$, and $h_{1}(u)=0$ otherwise.  We define edge costs $w_1: E[G_i] \rightarrow \mathbb{R}_{\geq 0}$, with  $w_{1}(e)=c(e)$ for all $e\in E[G_i]$ that are not incident to $v_2$, and $w_1(e)=\infty$ for $e\in E[G_i]$ that are incident to $v_2$. We apply Theorem~\ref{alg:DP} to the subgraph induced on $V(G_i)$ with edge costs $w_{1}$, and requirement function $h_{1}$, and denote the resulting cost by $c_{1}^i$.
    
    \item For the case that $v_2$ is covered by links in $E[G_i]$ but $v_1$ is not, we define requirement function $h_2 : V(G_i) \rightarrow \{0,1\}$ by $h_{2}(v_1)=0$, and $h_{2}(v_2)=1$, and for $u \in V(G_i)\backslash\{v_1,v_2\}$, $h_{2}(u) = 1$ if $u$ is degree $3$, and $h_{2}(u)=0$ otherwise.  We define edge costs $w_2: E[G_i] \rightarrow \mathbb{R}_{\geq 0}$, with  $w_{2}(e)=c(e)$ for all $e\in E[G_i]$ that are not incident to $v_1$, and $w_2(e)=\infty$ for $e\in E[G_i]$ that are incident to $v_1$. We apply Theorem~\ref{alg:DP} to the subgraph induced on $V(G_i)$ with edge costs $w_{2}$, and requirement function $h_{2}$, and denote the resulting cost by $c_{2}^i$.
\end{enumerate}

We now add to $V'$ the nodes $v_1$ and $v_2$, and add the edges of $\delta(G_i)$ to $E'$, with costs $c'(v_1)= c(v_1)$ and $c'(v_2)=c(v_2)$. We also add nodes $u_1$ and $u_0$ to $V'$, and add edges $u_0u_1,$ $u_1v_1,$ $u_1v_2,$ $v_1v_2$ to $E'$ with costs $c'(u_0u_1) = c_0^i$, $c'(u_1v_1) = c_1^i$, $c'(u_1v_2) =c_2^i$, and, $c'(v_1v_2) = c_{1,2}^i - \min\{c_0^i,c_1^i,c_2^i\}$. (See Figure~\ref{fig:degree1}.(c) and (d)). 

Note that it is possible that $c'(v_1v_2)$ is negative, but as we will see, this will not cause any problems. We call this subgraph of $G'$ a degree 2 gadget, with costs $c_0^i$, $c_1^i$, $c_2^i$, and $c_{12}^i$ and add $u_1,v_1,$ and $v_2$ to $N$. Lastly, we add every degree 3 node of $V'$ to $N$. 

We now wish to prove the following lemma
\begin{lemma}
\label{lem:degree12reduction}
    For the $4$-Obstruction Covering instance $(G,4,L,c)$, there is a set of links $F\subseteq L$, that covers every degree 1 or 2 $\{G_1,\dots, G_k\}$ and every degree 3 with minimum cost $c^*$ if and only if for the corresponding $N$-Edge Cover instance, $(G',c')$, described above has an optimal solution $F' \subseteq E'$ with cost $c^*$.
\end{lemma}
\begin{proof}
    First, given a covering $F$ of degree $3$ nodes and degree 1 or 2 ladders and hexagons of $G$ with minimum cost $c^*$, we will find an $N$-edge covering $F'$ of $G'$ with cost $c^*$. 

    First, for every link $\ell\in F$ that is not in a degree 1 or 2 $G_i\in \{G_1,\dots, G_k\}$, we add $\ell$ to our $N$-edge cover $F'$. Thus, $F'$ covers every degree 3 node in $G'$ that is not part of a gadget (and maybe some nodes that are) for the same cost $F$ takes to cover.

    Next, consider degree 1 or 2 $G_i\in \{G_1,\dots, G_K\}$. If $G_i$ is degree 1, then it is clear that $c(F\cup E[G_i]) = c_0^i $ if $F\cap E[G_i]$ does not cover the corner $v$ of $G_i$, and $c(F\cup E[G_i]) =  c_1^i\}$ if it does cover $v$. This follows by the construction of the gadget and the definition of $c_0^i$ and $c_1^i$. Thus, $F'$ covers the nodes of the gadget with cost equal to that of $F$.

    If $G_i$ is degree 2, then we can easily see the following by construction of the gadget and the definition of $c_0^i$, $c_1^i$, $c_2^i$, and $c_{1,2}^i$. 
    \begin{enumerate}
        \item If $F\cap E[G_i]$ covers neither $v_1$ nor $v_2$, then $c(F\cup E[G_i]) = c_0^i$.
        \item If $F\cap E[G_i]$ covers $v_1$ but not $v_2$, then $c(F\cup E[G_i]) = c_1^i$.
        \item If $F\cap E[G_i]$ covers $v_2$ but not $v_1$, then $c(F\cup E[G_i]) = c_2^i$.
        \item If $F\cap E[G_i]$ covers both $v_1$ and $v_2$, then $c(F\cup E[G_i]) = c_{1,2}^i$.
    \end{enumerate}
    In cases $1-3$, we add to $F'$ the edge of the degree 1 gadget that has the corresponding cost. This covers $u_1$, and $F$ covers every degree $3$ node, so it will use links in $\delta(G_i)$ to cover whichever corner is not covered by $F\cap E[G_i]$. Such an edge is added to $F'$, so $v_1$ and $v_2$ will also be covered by $F'$. 

    In case $4$, we add the edge $v_1v_2$ to $F'$ and the edge of $\{u_0u_1,u_1v_1,u_1v_2\}$ with minimum cost, and by definition of $c'(v_1v_2)$, $F'$ has cost $c_{1,2}^i$ on the gadget. Thus, $F'$ covers the nodes of the gadget with cost equal to that of $F$.

    Therefore, we have $c'(F')=c(F)=c^*$. So it remains to show that given a solution $F'$ with minimum cost $c^*$, we can find a set of links $F\subseteq L$ with cost $c(F)=c^*$ that satisfies the claim.
    
    Given an $N$-edge covering $F'$ of $G'$ of minimum cost $c'(F')=c^*$, we will find a set of edges $F\subseteq E$ that cover the degree $3$ nodes of $G$ and the degree $1/2$ ladders/hexagons with cost $C^*$. We begin with $F=\emptyset$ and will add to $F$ step by step.

    First, for every edge $e\in F'$ that is not in a gadget we add $e$ to $F$. So $F$ covers every degree 3 node not in a degree 1 or 2 $G_i\in \{G_1,\dots, G_k\}$.

    We now show how to use $F'$ to add edges to $F$ that cover degree 1 or 2 $G_i\in \{G_1,\dots, G_k\}$.

    \paragraph{Case 1:}First consider the case where $G_i$ is degree 1, and the corresponding degree 1 gadget in $G'$ has costs $c_0^i$ and $c_1^i$. Let $v$ denote the corner of $G_i$ and denote the nodes of the corresponding gadget by $v,v_0,v_1$. We consider in which cases the edges $v_0v_1,v_1v$ are in $F'$.
    \begin{itemize}
        \item Case: $v_0v_1, v_1v\notin F'$. The node $v_1\in N$ is not covered by $F'$ and thus $F'$ is not a feasible $N$-edge cover.

        \item Case: $v_1v\in F'$. We add to $F$ the edges found by applying Theorem~\ref{alg:DP} to $E[G_i]$ with covering requirements $h$ and edge weights $w$. Thus, $F$ covers every square of $G_i$ and every degree $3$ node of $G_i$ with cost $c(G_i\cap F) = c_1^i = c'(v_1v)$.

        \item Case: $v_0v_1\in F'$. We add to $F$ the edges found by applying Theorem~\ref{alg:DP} to $E[G_i]$ with covering requirements $h'$ and edge weights $w'$. Thus, $F\cap E[G_i]$ covers every square of $G_i$ and every degree 3 node of $G_i$ except $v$ with cost $c(G_i\cap F) =c_0^i = c'(v_0v_1)$. Note that since $F'$ covers $v$ but does not contain $v_1v$, it must contain the edge in $\delta(G_i)$, which is incident to $v$ and is in $F$ by construction. 

        \item Case: $v_0v_1,v_1v\in F'$. Since $v_0\notin N$, we can conclude that at least one of $c_0^i$ or $c_1^i$ is equal to $0$. Thus, we can remove $v_0v_1$ from $F'$ for no increase in cost or loss of feasibility and apply Case 2. 
    \end{itemize}
    Thus, for any degree 1 $G_i\in \{G_1,\dots, G_k\}$, $F$ covers the squares of $G_i$ and the degree 3 nodes of $G_i$ with the same cost as $F'$ cover the corresponding degree 1 gadget.
    \paragraph{Case 2:}Now, suppose that $G_i$ is degree 2, and the corresponding degree 2 gadget in $G'$ has costs $c_0^i,c_1^i,c_2^i,$ and $c_{12}^i$. Denote by $v_1, v_2\in V(G_i)$ the corners of $G_i$. We consider in which cases the edges $u_0u_1,u_1v_1,u_1v_2,$ and $v_1v_2$ are in $F'$.

    \begin{itemize}
        \item Case: $\delta(u_1)\cap F' = \{v_1v_2\}$. We add to $F$ the edges found by applying Theorem~\ref{alg:DP} to $E[G_i]$ with covering requirements $h_{12}$ and edge weights $w_{12}$. Thus, $F$ covers every square and degree 3 node of $G_i$ with cost $c_{12}^i$. Note that since $u_1\in N$, $F'$ contains an edge of $\delta(u_1)$ of least cost, the cost of $F'$ on the gadget is $c'(v_1v_2) + \min\{u_1v_1,u_1v_2,u_0u_1\} = c_{12}^i$. For the remaining cases, we assume $v_1v_2\notin F'$.

        \item Case: $\delta(u_1)\cap F' = \{u_1v_1\}$. We add to $F$ the edges found by applying Theorem~\ref{alg:DP} to $E[G_i]$ with covering requirements $h_1$ and edge weights $w_1$. Thus, $F\cap E[G_i]$ covers every square of $G_i$, every degree 3 node in $V(G_i)$ including $v_2$, but not $v_1$, with cost $c(G_i\cap F) = c_1^i$. Note that since $F'$ covers $v_1$ and $v_2$, but the gadget does not contain $v_1v_2$ or $u_1v_2$, that $F'$ must contain the edge of $\delta(G_i)$ that is incident to $v_2$, and thus $F$ also contains that edge. Thus, $F$ covers every degree $3$ node and every square of $G_i$. 

        \item Case: $\delta(u_1)\cap F' = \{u_1v_2\}$. This case can be handled similarly to Case 2.

        \item Case: $\delta(u_1)\cap F' = \{u_0u_1\}$. We add to $F$ the edges found by applying Theorem~\ref{alg:DP} to $E[G_i]$ with covering requirements $h_0$ and edge weights $w_0$. Thus, $F\cap E[G_i]$ covers every square of $G_i$ and every degree 3 node except $v_1$ and $v_2$. Note that since $F'$ covers $v_1$ and $v_2$, it must contain both edges of $\delta(G_i)$, and thus by construction $F$ contains those edges as well and thus covers every degree 3 node of $G_i$.
            
        \item Case: $|\delta(u_1)\cap F'|\geq 2$. If $u_0u_1\in F'$, then we can remove it from $F'$ for no loss of feasibility, and no increase in the cost of $F'$. If, after removing $u_0u_1$, $|\delta(u_1)\cap F'|=1$, then $\delta(u_1)\cap F' \in \{u_1v_1,u_1v_2\}$ and we can consider either Case 2 or 3. So assume that $\delta(u_1)\cap F' = \{u_1v_1,u_1v_2\}$. Without loss of generality assume that $c_1^i \leq c_2^i$. 

        It is clear that $c_{12}^i \leq c_1^i+c_2^i$, and thus we can replace $u_1v_2$ with $v_1v_2$ for a change in the cost of $F'$ on the degree 2 gadget from $c'(u_1v_1) + c'(u_1v_2) = c_1^i + c_2^i$ to $c'(u_1v_1) + c'(v_1v_2) =  c_1^i + c_{12}^i - \min\{c_0^i,c_1^i,c_2^i\} \leq c_{12}^i \leq c'(u_1v_1) + c'(u_1v_2)$. Clearly, the feasibility of $F'$ is maintained, and we can apply Case 1 to this gadget.
    \end{itemize}
    Thus, for any degree 2 $G_i\in \{G_1,\dots, G_k\}$, $F$ covers the squares of $G_i$ and the degree 3 nodes of $G_i$ with the same cost as $F'$ cover the corresponding degree 2 gadget.
    
    Therefore, given minimum $N$-edge cover $F'$ of $G'$ of cost $c^*$, we can find edges $F\subseteq E$ of $G$ that cover every degree 3 node and every square of degree 1 or 2 $G_i\in \{G_1,\dots, G_k\}$.
\end{proof}

The proof of Lemma~\ref{lem:degree12edgecover} follows directly from Lemma~\ref{lem:degree12reduction}.


    \section{ Missing proofs of Section~\ref{sec:4/3approx}}
This section includes the missing proofs of Section~\ref{sec:4/3approx}.

\subsection{ Proof of Lemma~\ref{lem:degree123}}
\label{sec:degree123}
    This proof follows directly from Lemma~\ref{lem:degree12reduction} where $G$ has edge costs $c(e)=1$ if $e\in L$ and $c(e)=\infty$ if $e\notin L$.

\subsection{Proof of Lemma~\ref{lem:simplifiedapproximation}}
\label{sec:simplifiedapproximationproof}
    We begin by defining notation. For a square $S$, we call a corner $u\in V(S)$ a \emph{good} node if there is a corner $v\in V(S)$ such that $uv\in L$, and we say $u$ is a bad node otherwise (recall that a node of $V(S)$ is a corner if it is incident to an edge in $\delta(S)$). Note that the nodes of $V(S)$ that are not corners are neither good nor bad.
    Note that since no links are necessary, all nodes of a degree 4 ladder are good nodes. 

    We find $APX$ by performing a sequence of steps. We begin with step 1 where we set $APX_1 = EC_3$, $Y_1= \emptyset$, and $\cH_1$ equal to the subset of $\{G_1,\dots, G_k\}$ whose squares are covered by $APX_1$. For each step $i$, $i \ge 2$, we define
    \begin{itemize}
    \item  $APX_i\subseteq L$ to be a partial solution;
     \item $\cH_i\subseteq \{G_1,\dots, G_k\}$ to be the set of ladders whose squares are covered by $APX_i$; 
     \item $Y_i$ to be the subset of edges  $ \subseteq APX_i \cap \delta(G[\cH_i])$ such that every $y\in Y_i$ has an endpoint in $R$.
    \end{itemize}
    Our goal is to compute $APX_i$ at $i$-th step, $i\ge 2$, such that the following holds: 
    \begin{enumerate}
        \item[(1)] The obstructions covered by $APX_{i-1}$ are a strict subset of those covered by $APX_i$;
        \item[(2)] $|APX_i\cap (Y_i \cup E[\cH_i])| \leq \frac{4}{3}| EC_3\cap (Y_i \cup E[\cH_i])|$, and $APX_i \backslash (Y_i \cup E[\cH_i]) = EC_3 \backslash (Y_i \cup E[\cH_i])$. 
    \end{enumerate}
    Note that since $APX_i$ covers the obstructions covered by $APX_{i-1}$, we must have $\cH_{i-1}\subseteq \cH_i$. If in iteration $i$, we have $\cH_i = \{G_1,\dots, G_k\}$, then we are done as $APX_i$ covers every obstruction. If $\cH_i \subsetneq \{G_1,\dots, G_k\}$, then we will find some $G_s\in \{G_1,\dots, G_k\} \backslash \cH_i$, and using $G_s$, we will find $APX_{i+1}$, $\cH_{i+1}$, and $Y_{i+1}$.
    
    We will consider three cases for $G_s$, for which we now provide a brief description. First, if $G_s$ has two nodes that are adjacent to  nodes in $R$ or in $\cH_i$. Second, if a good node of $G_s$ is adjacent to a good node of some $G_t\notin \cH_i$. And finally, if neither of the previous cases hold, then we show $G_s$ is degree 3 and that a good node of $G_s$ must be adjacent to a bad node of some degree 3 ladder $G_t\notin \cH_i$. In all cases, we will be able to find at least three links in $\delta(G_s)$ (and possibly $\delta(G_t)$), such that, by adding links to cover $G_s$ (and possibly $G_t$), we can charge the addition of these new links to the links in $\delta(G_s)$ (and possibly $\delta(G_t)$) for a $\frac{4}{3}$ fractional increase.

    \textit{Case 1:} there exists $G_s \in \{G_1,\dots, G_k\}\backslash \cH_i$, such that $G_s$ has two nodes that are adjacent to either a node in $R$ or in $\cH_i$.
    Since $G_s$ is degree at least $3$, there are links $\{\ell,\ell_1,\ell_2\} \subseteq \delta(G_s)\cap APX_i$. Without loss of generality, let the links $\ell_1$ and $\ell_2$ be the ones incident to nodes in $\cH_i$ or in $R$.
    
    Since no links are necessary, there is a link $e\in E(G_s)$ that shares an endpoint with $\ell$, and we start with $APX_{i+1}\leftarrow APX_{i}\cup \{e\}$. 
    Note that this ensures that $G_s$ is in $\cH_{i+1}$, and so $\cH_i \subsetneq \cH_{i+1}$.
    
    If any of $\ell_1$, $\ell_2$ or $\ell$ are incident to a node in $R$, then that link is in $Y_{i+1}$. Assume $\ell$ is incident to some ladder $G_t$. If $G_t\in \cH_i$, we are done this iteration. If not, we are going to change $APX_{i+1}$ by swapping two edges. Since $\ell$ is not necessary, there is link $f\in E(G_t)$ sharing an endpoint with $\ell$. We replace $\ell$ with $f$ in $APX_{i+1}$, and so also $G_t$ will be in $H_{i+1}$. 
    Since the endpoints of $\ell$ are covered in $APX_{i+1}$ by either $\ell$ or $\{e,f\}$, the obstructions in $G_s$ are covered, and all other obstructions of $APX_i$ are covered by $APX_{i+1}$, $(1)$ holds  (See Figure~\ref{fig:proofsketch}).
    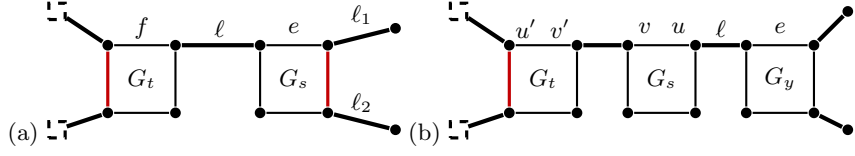
\begin{figure}[t]
    \begin{center}
        \begin{tabular}{c c}
            (a)
            \begin{tikzpicture}[scale=0.9]
                
                \tikzset{black dot/.style={draw=black, very thick, circle,minimum size=0pt, inner sep=1pt, outer sep=1pt,fill=black}}
                \tikzset{terminal/.style={draw=black,  thick,minimum size=0pt, inner sep=2.5pt, outer sep=1pt}}
                \tikzset{P node/.style={fill={rgb,255: red,20; green,154; blue,0}, draw={rgb,255: red,20; green,154; blue,0}, circle, minimum size=0pt,inner sep=1pt, outer sep=1pt}}
                \tikzset{Writing/.style={shape=circle} }
            
                \tikzstyle{witness edge}=[-, draw={rgb,255: red,195; green,0; blue,3}, very thick]
                \tikzstyle{T edges}=[-, thick]
                \tikzstyle{Fat edge}=[-, ultra thick]
                \tikzstyle{new witness}=[-, draw={rgb,255: red,195; green,0; blue,3}, dashed, ultra thick]
                \tikzstyle{connected terminals}=[-, draw=black, dashed, very thick]
                \tikzstyle{P}=[-, draw={rgb,255: red,20; green,154; blue,0}, very thick]
                
    		\node [style=black dot] (1) at (0, 0) {};
    		\node [style=black dot] (2) at (0, 1) {};
    		\node [style=black dot] (3) at (1, 1) {};
    		\node [style=black dot] (4) at (1, 0) {};
    		\node [style=black dot] (5) at (2, 1.25) {};
    		\node [style=black dot] (6) at (2, -0.25) {};
    		\node [style=black dot] (7) at (-1.25, 1) {};
    		\node [style=black dot] (8) at (-1.25, 0) {};
    		\node [style=black dot] (9) at (-2.25, 0) {};
    		\node [style=black dot] (10) at (-2.25, 1) {};
    		\node [style=connected terminals] (11) at (-3, 1.5) {};
    		\node [style=connected terminals] (12) at (-3, -0.25) {};
                \node (13) at (0.5, 0.5) {$G_s$};
        		\node (14) at (-1.75, 0.5) {$G_t$};
    		\node (16) at (-0.6, 1.25) {$\ell$};
                \node (16) at (0.5, 1.25) {$e$};
                \node (16) at (-1.75, 1.25) {$f$};
                  \node (16) at (1.5, 1.45) {$\ell_1$};
                  \node (16) at (1.5, 0.15) {$\ell_2$};
    		\node (17) at (-1, 1.25) {};
    		\draw [style=T edges] (4) to (1);
    		\draw [style=T edges] (1) to (2);
    		\draw [style=T edges] (7) to (8);
    		\draw [style=T edges] (8) to (9);
    		\draw [style=T edges] (10) to (7);
    		\draw [style=witness edge] (9) to (10);
    		\draw [style=witness edge] (3) to (4);
    		\draw [style=T edges] (2) to (3);
    		\draw [style=Fat edge] (4) to (6);
    		\draw [style=Fat edge] (3) to (5);
    		\draw [style=Fat edge] (2) to (7);
    		\draw [style=Fat edge] (10) to (11);
    		\draw [style=Fat edge] (9) to (12);
            \end{tikzpicture}
            &
            (b)
            \begin{tikzpicture}[scale=0.9]
                
                \tikzset{black dot/.style={draw=black, very thick, circle,minimum size=0pt, inner sep=1pt, outer sep=1pt,fill=black}}
                \tikzset{terminal/.style={draw=black,  thick,minimum size=0pt, inner sep=2.5pt, outer sep=1pt}}
                \tikzset{P node/.style={fill={rgb,255: red,20; green,154; blue,0}, draw={rgb,255: red,20; green,154; blue,0}, circle, minimum size=0pt,inner sep=1pt, outer sep=1pt}}
                \tikzset{Writing/.style={shape=circle} }
            
                \tikzstyle{witness edge}=[-, draw={rgb,255: red,195; green,0; blue,3}, very thick]
                \tikzstyle{T edges}=[-, thick]
                \tikzstyle{Fat edge}=[-, ultra thick]
                \tikzstyle{new witness}=[-, draw={rgb,255: red,195; green,0; blue,3}, dashed, very thick]
                \tikzstyle{connected terminals}=[-, draw=black, dashed, very thick]
                \tikzstyle{P}=[-, draw={rgb,255: red,20; green,154; blue,0}, very thick]
                
		\node [style=black dot] (1) at (-0.5, 0) {};
		\node [style=black dot] (2) at (-0.5, 1) {};
		\node [style=black dot] (3) at (0.5, 1) {};
		\node [style=black dot] (4) at (0.5, 0) {};
		\node [style=black dot] (5) at (1.25, 1) {};
		\node [style=black dot] (7) at (-1.25, 1) {};
		\node [style=black dot] (8) at (-1.25, 0) {};
		\node [style=black dot] (9) at (-2.25, 0) {};
		\node [style=black dot] (10) at (-2.25, 1) {};
		\node [style=connected terminals] (11) at (-3, 1.5) {};
		\node [style=connected terminals] (12) at (-3, -0.25) {};
		\node [style=Writing] (13) at (-1.75, 0.5) {$G_t$};
		\node [style=Writing] (14) at (0, 0.5) {$G_s$};
		\node [style=Writing] (16) at (1.75, 1.25) {$e$};
		\node [style=black dot] (18) at (1.25, 0) {};
		\node [style=black dot] (19) at (2.25, 0) {};
		\node [style=black dot] (20) at (2.25, 1) {};
		\node [style=Writing] (21) at (1.75, 0.5) {$G_y$};
		\node [style=Writing] (22) at (-1.5, 1.25) {$v'$};
		\node [style=Writing] (23) at (-2, 1.25) {$u'$};
		\node [style=Writing] (24) at (-0.25, 1.25) {$v$};
		\node [style=Writing] (25) at (0.25, 1.25) {$u$};
		\node [style=Writing] (26) at (0.875, 1.25) {$\ell$};
		\node [style=black dot] (27) at (2.75, 1.5) {};
		\node [style=black dot] (28) at (2.75, -0.25) {};
          
		\draw [style=T edges] (4) to (1);
		\draw [style=T edges] (1) to (2);
		\draw [style=T edges] (7) to (8);
		\draw [style=T edges] (8) to (9);
		\draw [style=T edges] (7) to (10);
		\draw [style=witness edge] (9) to (10);
		\draw [style=T edges] (3) to (4);
		\draw [style=Fat edge] (2) to (7);
		\draw [style=Fat edge] (10) to (11);
		\draw [style=Fat edge] (9) to (12);
		\draw [style=T edges] (2) to (3);
		\draw [style=Fat edge] (3) to (5);
		\draw [style=T edges] (5) to (20);
		\draw [style=T edges] (20) to (19);
		\draw [style=T edges] (19) to (18);
		\draw [style=T edges] (18) to (5);
		\draw [style=Fat edge] (20) to (27);
		\draw [style=Fat edge] (28) to (19);

            \end{tikzpicture}
        \end{tabular}
    \end{center}
    \caption{In both figures $APX_i$ is shown as bold edges. The red edges are not links.\\ (a) Ladder $G_s$ adjacent to two lonely nodes and ladder $G_t$.  We can see $\{\ell,\ell_1,\ell_2\} = \delta(G_s)$. We can add $f$ and replace $\ell$ with $e$, charge the addition of $f$ to $\ell,\ell_1,\ell_2\}$ for an increase from $3$ links to $4$, but covering both $G_s$ and $G_t$.
    (b) Ladders $G_s$ and $G_t$ with adjacent good nodes $v$ and $v'$. We add $uv$ and $u'v'$ to $APX_{i+1}$ and remove $vv'$, and we replace $\ell$ with $e$. Charging the addition of $vv'$ and $e$ to $\ell,vv',$ and the edge incident to $u'$.}
    \label{fig:proofsketch}
\end{figure}
    By constructions of $APX_i$ and $APX_{i+1}$, we have 
    \begin{align*}
        APX_{i+1}\backslash (Y_{i+1} \cup E[\cH_{i+1}] )  &= APX_{i}\backslash (Y_{i+1} \cup E[\cH_{i+1}])=  EC_3\backslash (Y_{i+1} \cup E[\cH_{i+1}] ) 
    \end{align*}
     
    To see that $|APX_{i+1}\cap (Y_{i+1} \cup E[\cH_{i+1}])| \leq \frac{4}{3}| EC_3\cap (Y_{i+1} \cup E[\cH_{i+1}])|$, observe that 
    \begin{align*}
        |EC_3\cap (Y_{i+1}\cup E[\cH_{i+1}])| - |EC_3\cap (Y_i\cup E[\cH_i])|\coloneqq \Delta \ge 3
    \end{align*}
    where the inequality follows since $\ell,\ell_1,\ell_2\notin EC_3\cap(Y_i\cup E[\cH_i])$, but are in $EC_3\cap(Y_{i+1}\cup E[\cH_{i+1}])$. Similarly, the following inequality follows since $APX_{i+1}\cap(Y_{i+1}\cup E[\cH_{i+1}])$ contains $\ell_1,\ell_2,e$ and exactly one of $f$ or $\ell$.
    \begin{align*}
        |APX_{i+1}\cap (Y_{i+1}\cup E[\cH_{i+1}])| - |APX_i\cap (Y_i\cup E[\cH_i])| = \Delta +1
    \end{align*}
    So as $i$ increases by one, the increase on the right hand side of the inequality of (2) is at most $\frac{4}{3}$ the increase on the left hand side. Thus, (2) holds.


    \paragraph{Case 2:}
    there exist $G_s, G_t \in \{G_1,\dots, G_k\}\backslash \cH_i$ such that a good node of $G_s$ is adjacent to a good node of $G_t$.  
    Let $u,v\in G_s$ and $u',v'\in G_t$ be good nodes such that $vv'\in L$. Since $v$ and $v'$ are degree $3$ nodes that are covered by $APX_i$, we see $vv'\in APX_i$. We set $APX_{i+1} \leftarrow APX_{i} \cup \{uv,u'v'\} \backslash \{vv'\}$.
    In this case, we see that 
    $G_s$ and $G_t$ are in $\cH_{i+1}$, and so  $\cH_i \subsetneq \cH_{i+1}$. 
    
    Since $u$ and $u'$ are good nodes, they are incident to links $\ell\in \delta(G_s)$ and $\ell'\in \delta(G_t)$ respectively. If $\ell$ (resp. $\ell'$) is incident to a node in $R$, then  $\ell$ (resp. $\ell'$) is in $Y_{i+1}$. Assume $\ell$ is incident to some $G_x$ (the case for $\ell'$ can be handled similarly if $\ell'$ is incident to some $G_y$). If $G_x = G_t$ (similarly if  $G_y=G_s$), then the endpoint of $\ell$ in $G_t$ cannot be adjacent to $v'$ or else $G[G_t \cup G_s]$ forms a ladder of length $3$, contradicting our assumption on $G$. 
        
    If $G_x\in \cH_i$, then $G_x$ is covered by $APX_i$ and we are done. If $G_x\notin \cH_i$, then $G_x$ is not covered by $APX_i$, and since no links are necessary, there is a link $e\in E(G_x)$ that is adjacent to $\ell$. We remove $\ell$ from $APX_{i+1}$ and add $e$, and so $G_x$ is in $\cH_{i+1}$. Since both endpoints of $\ell$ are covered by $uv$ and $e$, $APX_{i+1}$ covers the obstructions covered by $APX_i$ as well as the obstructions in $G[\cH_{i+1}\cup Y_{i+1}]$. (See Figure~\ref{fig:proofsketch}.)
    
  By a similar argument to the first case, we can see that (1) and (2) hold. 
  
    \paragraph{Case 3:}
    Neither of the above cases hold. In this case we wish to show that every $G_s\in \{G_1,\dots, G_k\} \backslash \cH_i$ is degree $3$, and that there exists a $G_t\in \{G_1,\dots, G_k\} \backslash \cH_i$ such that $G_t$ has a good node that is adjacent to a bad node of $G_s$.
    First, recall that every node of a degree 4 ladder is a good node. 
    
    First, we show that the ladders $G_s$ and $G_t$ described always exist when cases 1 and 2 do not hold. We denote by $X_3$ the set of degree $3$ ladders in $\{G_1, \dots, G_k\}\backslash \cH_i$. Denote by $N$ the set of links between a good node and a bad node of different ladders in $\{G_1,\dots, G_k\} \backslash \cH_i$. The following claim will show that the degree 3 ladders that are not covered by $APX_i$ has the same size as $N$.
    \begin{claim}
        $|X_3| = |N|$
    \end{claim}
    \begin{proof}
       In order to show that $|X_3|\ge |N|$, we show that every ladder of degree at least three has at most one bad node. This follows from the fact that as there exists no necessary links, then at least two of the three edges incident to every corner of a square must be links. Therefore, in a ladder of degree four all corners are good, and in a ladder of degree three at most one corner is bad.
            
        To see $|X_3| \leq |N|$, we first note since case 1 does not hold, a ladder in $\{G_1, \dots, G_k\}\backslash\cH_i$ is adjacent to at most one of: a lonely node, or, a ladder in $\cH_i$. Secondly, since case 2 does not hold, the good nodes of ladders in $X_3$ are not adjacent to the good nodes of any ladder in $\{G_1, \dots, G_k\}\backslash\cH_i$. Therefore, every degree 3 ladder in $\{G_1, \dots, G_k\}\backslash\cH_i$ has at least one good node adjacent to a bad node of another degree 3 ladder, and thus $|X_3| \leq |N|$. And the claim is proven.
    \end{proof}
        
    Therefore, since $|X_3|=|N|$, one of the good nodes of each degree 3 ladder in $X_3$ is adjacent to a bad node of a ladder in $X_3$, and furthermore, the other good node must be adjacent to a node in $R$ or a ladder in $\cH_i$. Since each bad node is adjacent to a degree 3 ladder, the good nodes of degree 4 ladders can only be adjacent to either good nodes, nodes in $R$, or elements of $\cH_i$, and thus must be covered by $APX_i$ since neither case (1) nor (2) hold. Therefore, $X_3=\{G_1, \dots, G_k\}\backslash\cH_i$.
    
    We now consider degree 3 ladder $G_s\in \{G_1,\dots, G_k\} \backslash \cH_i$ with bad node $s_1$, and good nodes $s_2, s_3$, where $s_3$ is adjacent to a node in $R$ or a node in $\cH_i$. And let $s_4$ be the node of $G_s$ that is not a corner.  Note that $s_1s_4$ must be a link since there are no necessary links and $s_1s_2$ cannot be a link by the definition of bad nodes.
    
    By the argumentation above, there is a degree 3 ladder $G_t\in\{G_1,\dots, G_k\} \backslash \cH_i$ with bad node $t_1$, and good nodes $t_2, t_3$, where $t_3$ is adjacent to a node in $R$ or a node in $\cH_i$, such that $s_1t_2\in APX_i$. Let $e \in \delta(G_s)\cap APX_i$ be incident to $s_3$, and $f\in \delta(G_t)\cap APX_i$ be incident to $t_3$. We let $APX_{i+1} \leftarrow APX_i \cup \{s_1s_4,t_2t_3\}\backslash \{s_1t_2\}$, and so $G_s$ and $G_t$ are in $\cH_{i+1}$, and $\cH_{i}\subsetneq \cH_{i+1}$. The link $e$ (resp. $f$) is in $Y_{i+1}$ if $e$ (resp. $f$) is incident to a node in $R$.
    By a similar argument to  the first case we can see that $APX_{i+1}$, $Y_{i+1}$, and $\cH_{i+1}$ satisfy (1) and (2).    

    Note that for every $i\geq 1$, by (2) we have
    \begin{align*}
        |APX_i| &= |APX_i\cap (Y_i \cup E[\cH_i])| + |APX_i \backslash (Y_i \cup E[\cH_i])| \\
        &\leq \frac{4}{3}| EC_3\cap (Y_i \cup E[\cH_i])| + |EC_3 \backslash (Y_i \cup E[\cH_i])| \leq \frac{4}{3}|EC_3|
    \end{align*}
    Furthermore, in each iteration we cover a new element of $\{G_1,\dots, G_k\}$, so after a polynomial number of iterations we will have $\cH_i = \{G_1,\dots, G_k\}$ and we have covered every obstruction.

\subsection{ Proof of Theorem~\ref{thm:unweighted}}
\label{sec:4/3approx-proofs}
\label{sec:unweighted}
We apply Theorem~\ref{thm:ChainDecomposition} to decompose instance $(G,4,L)$ into hexagons and ladders $G_1,\dots, G_k$, and lonely nodes $R$. 

We first find a partial solution $EC_3$ as described in Section~\ref{sec:4/3approx} that covers degree 3 nodes and degree 1 or 2 $G_i\in \{G_1,\dots,G_k\}$ by applying Lemma~\ref{lem:degree123}. We next prove Lemma~\ref{lem:uncovered} which shows that with $EC_3$ we can find another partial solution $APX_1$ that covers every $G_i\in\{G_1,\dots, G_k\}$ with degree at least 2 by increasing the number of links on $G_i$ by a $\frac{4}{3}$ fraction, leaving the solution equal elsewhere. Finally, we can apply Lemma~\ref{lem:simplifiedapproximation} to cover the remaining $G_i\in \{G_1,\dots, G_k\}$ that contain obstructions not covered by $APX_1$ to find a feasible solution $APX_2$. Then we will show that $APX_2$ has cost at most $\frac{4}{3}c|opt|$.
\begin{restatable}{lemma}{uncovered}
\label{lem:uncovered}
    Consider an unweighted $4$-Obstruction Covering instance $(G,4,L)$, that contains no ladder or hexagon of degree $1$ or $2$, with subgraphs $R,G_1,\dots, G_k$, found by Theorem~\ref{thm:ChainDecomposition}. Consider partial solution $EC_3 \subseteq L$ that covers every degree $3$ node of $G$, but every $G_s\in \{G_1, \dots, G_k\}$ contains an obstruction that is not covered by $EC_3$.
    We can find in polynomial time, the following:
    \begin{enumerate}
        \item  A (potentially infeasible) solution $APX\subseteq L$;
        \item  A subset $\cH \subseteq \{G_1,\dots, G_k\}$  containing all hexagons and ladders of length at least $2$, such that $APX$ covers every obstruction in $G[\cH]$, and;
        \item A subset $Y \subseteq APX \cap \delta(G[\cH])$, such that, every $y\in Y$ has an endpoint in $R$
    \end{enumerate}
    such that, $|APX \cap (Y \cup E[\cH])| \leq \frac{4}{3}| EC_3 \cap (Y \cup E[\cH])|$, and $APX\backslash (Y \cup E[\cH]) = EC_3 \backslash (Y \cup E[\cH])$.
\end{restatable}
\begin{proof}
    We find $APX$ by performing a sequence of steps. We begin with step 1 where we set $APX_1=EC_3$, $Y_1=\emptyset$, and $\cH_1 = \emptyset$. For each step $i$, $i\geq 2$, we define
    \begin{itemize}
        \item $APX_i\subseteq L$ to a be a partial solution;
        \item $\cH_i\subseteq\{G_1,\dots, G_k\}$ contains the set of ladders and hexagons whose obstructions are covered by $APX_i$, as well as the the set of ladders and hexagons that contain an obstruction covered by $APX_{i}$ but not covered by $APX_{i-1}$;
        \item $Y_i$ to be the subset of edges $\subseteq APX_i\cap \delta(G[\cH_i])$ such that every $y\in Y_i$ has an endpoint in $R$.
    \end{itemize}
    Our goal is to compute $APX_i$ at the $i$-th step, $i\ge 2$, such that the following holds:
    \begin{enumerate}
        \item[(1)] The obstructions covered by $APX_{i-1}$ are a strict subset of those covered by $APX_{i}$;
        \item[(2)] $|APX_i\cap (Y_i \cup E[\cH_i])| \leq \frac{4}{3}| EC_3\cap (Y_i \cup E[\cH_i])|$, and $APX_i \backslash (Y_i \cup E[\cH_i]) = EC_3 \backslash (Y_i \cup E[\cH_i])$. 
    \end{enumerate}
    Note that since $APX_i$ covers the obstructions covered by $APX_{i-1}$ we have $\cH_{i-1}\subseteq \cH_{i}$, and since $EC_3=\cH_1$ already covers every degree 3 node, we must have $\cH_{i-1} \subsetneq \cH_i$. 
    If in iteration $i$,  every obstruction in a ladder or hexagon $G_s\in \cH_i$ is covered, and $\{G_1,\dots,G_k\}\backslash \cH_i$ contains no hexagons or ladders of length at least $2$, then we define $\cH \coloneqq \{G_s\in \cH_i| \text{every obstruction of } G_s \text{ is covered} \}$, $APX \coloneqq APX_i$,  and $Y\coloneqq Y_i$ and we are done. Otherwise, we will find a  hexagon or ladder of length at least 2,  and use it to find $APX_{i+1}, \cH_{i+1}$, and $Y_{i+1}$.
    
    Let $G_s\in \{G_1,\dots,G_k\}\backslash \cH_i$ be a hexagon or ladder of length at least 2. We will consider the cases for $G_s$. First, we will deal with the cases that $G_s$ contains a hexagon. Then, we will assume that $\{G_1,\dots,G_k\}\backslash \cH_i$ contains no hexagons, and will consider the case where $G_s$ is a ladder of length at least $2$.
    
    In both cases, for every link we add, we will be able to find at least three links in $\delta(G_s)$ (and possibly $\delta(G_t))$), such that, by adding links to cover $(G_s)$ (and possibly $(G_t)$), we can charge the addition of these new links to the links in $\delta(G_s)$ (and possibly $\delta(G_t)$) for a $\frac{4}{3}$ fractional increase.

    \paragraph*{Case: $G_s$ contains a hexagon.}
    First, suppose $G_s\in \{G_1,\dots, G_k\}\backslash \cH_i$  contains a hexagon.  By definition, a hexagon has exactly four degree 3 nodes whose neighbourhoods are in the hexagon. Moreover, the subgraph induced on these nodes forms a claw (see the nodes $\{u_1,u_2,u_2,v_0\}$ in Figure~\ref{fig:hexagoncovering}). Therefore, in order for $EC_3$ to cover each of the corresponding nodes in the hexagon of $G_s$, there must be at least three links $\ell_1,\ell_2,\ell_3\in EC_3\cap E[G_s]$. By definition of $APX_i$, we must have $\ell_1,\ell_2,\ell_3\in APX_i$. 
    Let $S \subseteq G_s$ denote the square that is not covered by $\cH_i$. Since no links are necessary, there is a link $\ell\in E[S]\cap L$. See Figure~\ref{fig:hexagoncovering} for a visual example of this step.
    \begin{figure}[t!]
    \begin{center}
        \begin{tabular}{c c}
             (a)
                \begin{tikzpicture}[scale=0.8]
            \tikzset{black dot/.style={draw=black, very thick, circle,minimum size=0pt, inner sep=1pt, outer sep=1pt,fill=black}}
            \tikzset{terminal/.style={draw=black,  thick,minimum size=0pt, inner sep=2.5pt, outer sep=1pt}}
            \tikzset{P node/.style={fill={rgb,255: red,20; green,154; blue,0}, draw={rgb,255: red,20; green,154; blue,0}, circle, minimum size=0pt,inner sep=1pt, outer sep=1pt}}
            \tikzset{Writing/.style={shape=circle} }
        
            \tikzstyle{witness edge}=[-, draw={rgb,255: red,195; green,0; blue,3}, very thick]
            \tikzstyle{T edges}=[-,  thick]
            \tikzstyle{new witness}=[-, draw={rgb,255: red,195; green,0; blue,3}, dashed, ultra thick]
            \tikzstyle{connected terminals}=[-, draw=black, dashed, very thick]
            \tikzstyle{P}=[-, draw={rgb,255: red,20; green,154; blue,0}, very thick]
            
    		\node [style=black dot] (0) at (0, 1.75) {};
    		\node [style=black dot] (1) at (-1.5, 0.75) {};
    		\node [style=black dot] (2) at (1.5, 0.75) {};
    		\node [style=black dot] (3) at (-1.5, -0.75) {};
    		\node [style=black dot] (4) at (1.5, -0.75) {};
    		\node [style=black dot] (5) at (0, -1.75) {};
    		\node [style=black dot] (6) at (0, 0) {};
    		\node [style=connected terminals] (7) at (2.5, 1.25) {};
    		\node [style=connected terminals] (8) at (0, -2.5) {};
    		\node [style=connected terminals] (9) at (-2.5, 1.25) {};
    		\node [style=Writing] (10) at (0.5, 1.75) {$v_1$};
    		\node [style=Writing] (11) at (0.5, 0) {$v_0$};
    		\node [style=Writing] (12) at (-2, -0.75) {$v_2$};
    		\node [style=Writing] (13) at (2, -0.75) {$v_3$};
    		\node [style=Writing] (15) at (0, -0.75) {$S$};
            \node [style=Writing] (16) at (-1.5, 1.25) {$u_1$};
    		\node [style=Writing] (17) at (1.5, 1.25) {$u_2$};
    		\node [style=Writing] (18) at (0.5, -1.75) {$u_3$};
    		\draw [style=T edges] (0) to (2);
    		\draw [style=T edges] (2) to (4);
    		\draw [style=P] (4) to (5);
    		\draw [style=T edges] (1) to (0);
    		\draw [style=T edges] (1) to (3);
    		\draw [style=P] (3) to (5);
    		\draw [style=T edges] (5) to (8);
    		\draw [style=T edges] (7) to (2);
    		\draw [style=T edges] (1) to (9);
    		\draw [style=T edges] (6) to (0);
    		\draw [style=P] (6) to (3);
    		\draw [style=P] (6) to (4);
    		\draw [style=new witness] (2) to (4);
    		\draw [style=new witness] (0) to (6);
    		\draw [style=new witness] (1) to (3);
        \end{tikzpicture}
             
             & 
             (b)
             \begin{tikzpicture}
            \tikzset{black dot/.style={draw=black, very thick, circle,minimum size=0pt, inner sep=1pt, outer sep=1pt,fill=black}}
            \tikzset{terminal/.style={draw=black,  thick,minimum size=0pt, inner sep=2.5pt, outer sep=1pt}}
            \tikzset{P node/.style={fill={rgb,255: red,20; green,154; blue,0}, draw={rgb,255: red,20; green,154; blue,0}, circle, minimum size=0pt,inner sep=1pt, outer sep=1pt}}
            \tikzset{Writing/.style={shape=circle} }
        
            \tikzstyle{witness edge}=[-, draw={rgb,255: red,195; green,0; blue,3}, very thick]
            \tikzstyle{T edges}=[-,  thick]
            \tikzstyle{new witness}=[-, draw={rgb,255: red,195; green,0; blue,3}, dashed, ultra thick]
            \tikzstyle{connected terminals}=[-, draw=black, dashed, very thick]
            \tikzstyle{P}=[-, draw={rgb,255: red,20; green,154; blue,0}, very thick]
            
		\node [style=black dot] (3) at (1, 1) {};
		\node [style=black dot] (4) at (1, 0) {};
		\node [style=black dot] (5) at (2, 1) {};
		\node [style=black dot] (6) at (2, 0) {};
		\node [style=black dot] (7) at (-1, 1) {};
		\node [style=black dot] (8) at (-1, 0) {};
		\node [style=black dot] (9) at (-2, 0) {};
		\node [style=black dot] (10) at (-2, 1) {};
		\node [style=connected terminals] (11) at (-2.75, 1.5) {};
		\node [style=connected terminals] (12) at (-2.75, -0.25) {};
		\node [style=connected terminals] (13) at (3, 1.5) {};
		\node [style=Writing] (14) at (-1.75, 1.25) {$v_1$};
		\node [style=Writing] (15) at (-0.75, 1.25) {$v_2$};
		\node [style=Writing] (17) at (1.25, 1.25) {$v_r$};
		\node [style=Writing] (18) at (2, 1.25) {$v_{r+1}$};
		\node [style=Writing] (19) at (-1.75, -0.25) {$u_1$};
		\node [style=Writing] (20) at (-0.75, -0.25) {$u_2$};
		\node [style=Writing] (22) at (1.25, -0.25) {$u_r$};
		\node [style=Writing] (23) at (2.25, -0.25) {$u_{r+1}$};

		\draw [style=T edges] (7) to (8);
		\draw [style=witness edge] (8) to (9);
		\draw [style=witness edge] (7) to (10);
		\draw [style=T edges] (9) to (10);
		\draw [style=T edges] (3) to (4);
		\draw [style=T edges] (10) to (11);
		\draw [style=T edges] (9) to (12);
		\draw [style=witness edge] (3) to (5);
		\draw [style=T edges] (4) to (6);
		\draw [style=T edges] (5) to (6);
		\draw [style=T edges] (5) to (13);
		\draw [style=connected terminals] (7) to (3);
		\draw [style=connected terminals] (8) to (4);
        \end{tikzpicture}
        \end{tabular}
    \end{center}
     \caption{(a) A hexagon, with uncovered square $S$ shown in green edges. The nodes $v_0,v_1,v_2$ and $v_3$ are degree three and must be covered by links that are in $E(S)$ (example links are shown with dashed red lines). To cover $S$ we can select an edge $\ell\in E(S)$ and add it to the solution, charging its addition to the dashed edges in $G_s$.
     \\
     (b) A ladder of length at least 2. Edges that are not links are shown in red, and edges that are links are shown in black. In this case, both of the squares $S_1$ and $S_r$ may be blockers depending on which links are in our initial partial solution, since the nodes $v_1,u_1,$ and $v_{r+1}$ are degree 3 but the edges $v_1v_2$, $u_1u_2$, and $v_rv_{r+1}$ are not links. }
     \label{fig:hexagoncovering}
\end{figure}
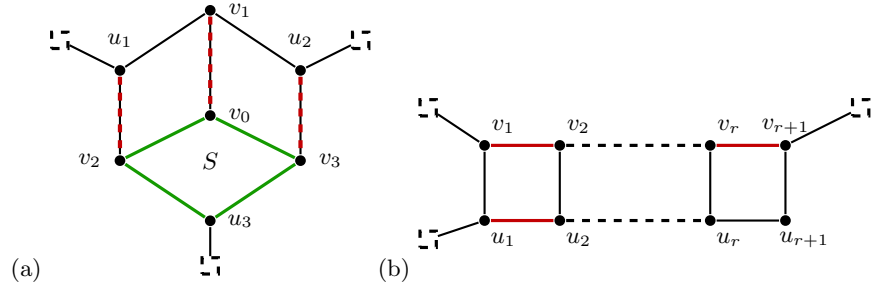
    
    We set $APX_{i+1} \leftarrow APX_{i}\cup \{\ell\}$. Therefore, $\cH_i\subsetneq \cH_{i+1}$, since $\cH_{i+1}$ contains $G_s$, and $APX_{i+1}$ satisfies condition (1). By construction, we have 
    \begin{align*}
        &|APX_{i+1} \backslash E[G_s]| = |APX_i \backslash E[G_s]|, \hspace{13pt} |APX_{i+1}\cap E[G_s]| = |APX_i\cap E[G_s]| + 1
    \end{align*}
    Therefore, we have that 
    \begin{align*}
        3 &\leq |EC_3 \cap (Y_{i+1}\cup \cH_{i+1})| - |EC_3 \cap (Y_{i}\cup \cH_{i})| \coloneqq \Delta \\
        4 &\leq |APX_{i+1} \cap (Y_{i+1}\cup \cH_{i+1})| - |APX_i \cap (Y_{i}\cup \cH_{i})| = \Delta +1
    \end{align*}
    So the left hand side of our desired inequality increases by at least $\Delta+1 $, while the left hand side increases by at least $\frac{4}{3}\Delta$. We can see that $\Delta \geq 3$ since $\ell,\ell_1,\ell_2\in EC_3$, so the desired inequality holds.
    
    \paragraph*{Case:  the set $\{G_1,\dots, G_k\} \backslash \cH_i$ contains no hexagons at all. } So $G_s$ is a ladder of length at least 2 with no hexagon.
    Applying Theorem~\ref{thm:ChainDecomposition}, let $G_s$ have consistent labelling $v_1,\dots, v_{r+1}, u_1,\dots, u_{r+1}$. Enumerate the squares of $G_s$ as $S_j = v_ju_ju_{j+1}v_{j+1}$, for $j=1,\dots,r$. Furthermore, we assume that if exactly one of $S_1$ and $S_r$ are not covered by $APX_i$, the consistent labelling of $G_s$ is given such that $S_r$ is uncovered. 
    
    To find $APX_{i+1}$, we will cover $G_s$ in a sequence of steps, iteratively guaranteeing that the squares of $G_s$ are covered one by one in increasing order of index. To do this, we will begin with a partial solution $F_0\coloneqq APX_i$ and will find a sequence of partial solutions $F_0,F_1,F_2,\dots$. 
    These solutions will satisfy the requirement that $F_\sigma$ covers the obstructions covered by $F_0,\dots, F_{\sigma-1}$, and  if $F_{\sigma-1}$ contains a square $S_j$ of least index that is not covered, then $F_\sigma$ will cover $S_j$. Therefore, after at most $r+1$ iterations, we will have computed $F_0,\dots, F_\sigma$, and $F_\sigma$ will cover every obstruction of $G_s$ and we can set $APX_{i+1} \coloneqq F_\sigma$.   
    
    We will show that if we must add a link $\ell$ to $F_\sigma$ we can either remove an edge from $APX_i$ or find a set of at least 3 links $\ell_1,\ell_2,\ell_3\in APX_i$ to charge the addition of $\ell$ to.  We guarantee the following about $\ell_1,\ell_2,\ell_3$ and $F_\sigma$;
    \begin{itemize}
        \item[(a)] $\ell_1,\ell_2,\ell_3\in EC_3$;
        \item[(b)] $\ell_1,\ell_2,\ell_3\in E(G_s)\cup \delta(G_s)$;
        \item[(c)] $\ell_1,\ell_2,\ell_3$ can not been have charged to before iteration $i$.
        \item[(d)] $|F_\sigma| \leq \frac{4}{3}|EC_3|$.
    \end{itemize}
    
    As previously stated, we begin with $F_0\coloneqq APX_i$. When we consider $F_\sigma$, if every square of $G_s$ is covered, then we let $APX_{i+1}\coloneqq F_\sigma$, and we are done. Otherwise, let $S_j$ be the square of least index that is not covered by $F_\sigma$. 

    Before we continue, we define some notation. For partial solution $F_\sigma$, we call the square $S_1\subseteq E(G_s)$ a \emph{blocker} of $F_\sigma$ when  
    a) $S_1$ is not covered by $F_\sigma$, (hence $S_r$ is not covered) and; 
    b) $v_1v_{2}\notin L$  if $v_1$ is degree $3$, and $u_1u_2\notin L$ if $u_1$ is degree $3$. Similarly, we call the square $S_r\subseteq E(G_s)$ a blocker of $F_\sigma$ when a) $S_r$ is not covered by $F_\sigma$, and; 
    c) $v_rv_{r+1}\notin L$ if $v_{r+1}$ if $v_{r+1}$ is degree 3, and $u_ru_{r+1}\notin L$ when $u_{r+1}$ is degree 3. See Figure~\ref{fig:hexagoncovering} for an example where $S_1$ and $S_r$ both  satisfy part (2) the definition of blockers.  
    
    Suppose that $S_1$ and $S_r$ both satisfy the part (b) of the definition of blockers for $F_\sigma$. We will show that if $r=2$ or $3$ that we can assume that at most one of $S_1$ and $S_r$  satisfy part (a) of the definition of blocker for $F_\sigma$, that is, that at most one is not covered by $F_\sigma$. For $r=2$, the nodes $v_2=v_r$ and $u_2=u_r$ are degree 3 and must be covered by $APX_i$, and thus by $F_\sigma$, therefore, at least one of $S_1$ or $S_2$ will contain a link and thus will be covered and cannot be a blocker. 
    
    If $r=3$, if $S_1$ and $S_2$ both satisfy the part (b) of the definition of blockers for $F_\sigma$, it is not hard to see that the edges $v_2u_2$ and $v_3u_3$ must b  e links,  since no links are necessary. If $S_1$ and $S_2$ satisfy part (a) of of the definition of blockers for $F_\sigma$, then for $F_\sigma$ to cover all degree 3 nodes while $S_1$ and $S_3$ are uncovered we must have $v_{2}v_{3},u_{3}u_{2}\in F_\sigma$ and thus in $L$. So we replace $F_\sigma$ with $F_\sigma  \cup \{ v_{2}u_{2}, v_{3}u_{3}\}\backslash \{v_{2}v_{3}, u_{2}u_{3}\}$.

    To show this, we consider the cases for $\sigma$: when $j=1$, when $j\in \{2,\dots, r-1\}$, and when $j=r$. 
    \paragraph{Case:} $j=1$. Note that in this case $\sigma=0$.
    Since $G_s$ is degree at least 3, we assume without loss of generality that our labelling is such that $v_1$  is incident to an edge $v_1w \in \delta(G_s)\cap F_0$. By our assumption on the labelling of $G_s$, it must be that both $S_1$ and $S_r$ are uncovered by $F_\sigma$. 
    
    If exactly one of $S_1$ and $S_r$ is a blocker, then we assume that our labelling is such that $S_1$ is not a blocker. We consider the cases  arising from whether $S_1$ is a blocker or not. 
    \begin{itemize}
        \item Case: $S_1$ not a blocker, thus at least one of $v_1v_2,u_1u_2\in L$. Assume without loss of generality that $v_1v_2\in L$. Set $F_1\leftarrow F_0 \cup \{v_1v_2\}$. 
        $EC_3$, and thus $F_0$ covers every degree $3$ node, so we have $u_2u_3,v_2v_3\in F_0\cap EC_3$. We charge the addition of $v_1v_2$ to $v_1w,v_2v_3,u_2u_3\in EC_3$.

        We now help ensure that $v_2v_3$ and $u_2u_3$ are not charged to again. If $r=2$, then $S_r$ is covered.
        
        If $r>2$ and $S_3$ is not covered by $F_0$.  Since no links are necessary, there is a link $\ell\in L\cap \delta(v_3)\cap E[S_3]$ that is not in $F_0$. We change $F_1\leftarrow F_1 \cup \{\ell\} \backslash \{v_2v_3\}$. Note that $F_1$ now covers both $S_1$ and $S_3$, in addition to the obstructions previously covered for no change in size, so we do not need to find links in $APX_i$ to charge $\ell$ to. 

        To help ensure that $v_1w\in \delta(G_s)\cap APX_i$ is charged to only once, consider the following cases.
        \begin{itemize}
            \item Case: if $w\in R$ or if $w\in G_x$ such that every obstruction of $G_s$ is covered, we are done with this iteration, and we see that  $v_1w$ is in $Y_{i+1}$. 
            \item Case: $w\in G_x$ for some $G_x\notin \cH_{i+1}$. Let $S\subseteq G_x$ be the square of $G_x$ that contains $w$. If $S$ is covered by $F_1$, we are done. If not, since no link is necessary, there is a link $\ell' \in E(S)\cap L$ that is adjacent to $v_1w$. We set $F_1\leftarrow F_1\cup \{\ell'\} \backslash\{v_1w\}$. So $F_1$ now covers both $S_1$ and $S$, in addition to the obstructions previously covered, for no increase in size. Thus, we do not need to find links in $EC_3$ to charge $\ell'$ to. We see that $G_x$ is in $\cH_{i+1}$. 
        \end{itemize}

        \item Case: $S_1$ is a blocker, thus $v_1v_2\notin L\cap E[S_1]$ (and possibly $u_1u_2\notin L\cap E[S_1]$). By our assumption on the labelling of $G_s$, both $S_1$ and $S_r$ are blockers, thus $r\geq 4$ by the discussion on blockers above. Furthermore, since $S_1$ is a blocker, $v_1u_1,v_2u_2$ must be in $E(S_1)\cap L$ since no links are necessary, and $v_2v_3,u_2u_3\in F_0$ since $F_0$ covers all degree 3 nodes. 
        We consider cases for which edges are in $F_0$ and $L$.        
        \begin{itemize}
            \item Case: $v_3u_3\in L\cap F_0$. We set $F_1\leftarrow  F_0 \cup \{v_2u_2\}\backslash \{u_2u_3\}$. Since $v_3u_3\in F_0$, the obstructions of $APX_i$ are still covered. 
            \item Case: $v_3u_3\in L\backslash F_0$. We set $F_1\leftarrow F_0 \{v_2u_2,v_3u_3\}\backslash\{v_2v_3,u_2u_3\}$. Obviously, all obstructions of $APX_i$ are still covered.
        \end{itemize}
        In either case, the size of $F_1$ is no more than the size of $F_0$ and we do not need to find links in $EC_3$ to charge $v_2v_3$ to.
        
        In the following cases, $v_3u_3\notin L$, and so $v_3v_4,u_3u_4\in L$ since no links are necessary.
        \begin{itemize}
            \item Case: at least one of $v_3v_4$ or $u_3u_4\in F_0$ (without loss of generality $v_3v_4\in F_0$). We set $ F_1 \leftarrow F_0 \cup \{v_2u_2\}\backslash \{v_2v_3\}$, so $F_1$ covers $S_1$ and $S_2$ with $v_2u_2$, in addition to the obstructions it previously covered, for no change in size, so we do not need to find links in $EC_3$ to charge $v_2u_2$ to. 

            \item Case: $v_3v_4,u_3u_4 \notin F_0$. We set $F_1 \leftarrow F_0 \cup \{v_2u_2,v_3v_4, u_3u_4 \}$ $\backslash \{v_2v_3, u_2u_3\}$. So we must find at least 3 links in $APX_i$ to charge one of $v_2u_2,v_3v_4, u_3u_4 $ to. We know that $r\geq 4$, and that both $v_4$ and $u_4$ are covered by $F_0$ so at least one of $v_4u_4,v_4v_5,u_4u_5$ must be in $F_0$.

            If $v_4u_4\in F_0$, then   we can charge $v_2u_2$ to $v_2v_3, u_2u_3, v_4u_4\in EC_3\cap APX_i$. Note that $F_1$ now covers $S_1$, $S_2$, $S_3$, and $S_4$. 

            If $v_4u_4\notin F_0$, then $v_4v_5,u_4u_5\in F_0$ since $v_4$ and $u_4$ are covered by $F_0$. In this case we charge $v_2u_2$ to $v_2v_3, u_2u_3, v_4v_5\in APX_i$. If $r< 6$ or $S_5$ is covered by $F_0$, then we are done with this iteration. So assume that $r\ge 6$, and $S_5$ is not covered by $F_0$. Then there must be a link $\ell\in E(S_5)\cap L\cap \delta(v_5)$. We set $F_1\leftarrow F_1 \cup \{\ell\}\backslash\{v_4v_5\}$.
        \end{itemize}
    \end{itemize}
    \paragraph{Case:} $j \in \{2,\dots, r-1\}$. $S_j$ is not covered, but $v_{j-1},v_j,u_{j-1},$ and $u_j$ are covered, since $APX_i$ covers every degree $3$ node. Therefore, it is clear that $u_{j-1}u_j,$ $v_{j-1}v_j,$ $u_{j+1}u_{j+2},$ $v_{j+1}v_{j+2} \in APX_i\cap F_\sigma$. Since no link is necessary, we know at least one of $v_jv_{j+1},$ $v_{j+1}u_{j+1}, u_ju_{j+1}\in L$. Let $\ell \in \{v_jv_{j+1},$ $v_{j+1}u_{j+1},$ $u_ju_{j+1}\}\cap L$ be a link from these three. Set $F_{\sigma+1} \leftarrow F_\sigma \cup \{\ell\}$, and charge $\ell$ to $u_{j-1}u_j,$ $v_{j-1}v_j,$ $u_{j+1}u_{j+2}\in APX_i$. 
    
    To help ensure that no links of $APX_i$ are charged twice, we consider the case $j<r-1$, and $S_{j+2}$ is not covered by $F_{\sigma+1}$. Then there is $\ell' \in E(S_{j+2})\cap L \cap \delta(u_{j+2})$. We set $F_{\sigma+1}\leftarrow F_{\sigma+1} \cup\{\ell'\} \backslash \{u_{j+1}u_{j+2}\}$, so $F_{\sigma+1}$ covers both $S_j$ and $S_{j+2}$, in addition to the obstructions it previously covered. If $r=j-1$, then every square of $G_s$ is covered by $F_{\sigma +1}$.
    
    \paragraph{Case:} $j=r$. Since $G_s$ is degree at least $3$, there is a $w\in V(G)$ such that $v_{r+1}w$ or $u_{r+1}w\in APX_i\cap \delta(G_s)$, we assume without loss of generality that  $u_{r+1}w\in APX_i \cap \delta(G_s)$. Since no link is necessary, there is a link $\ell\in S_{r}$ incident to $u_{r+1}$. Set $F_{\sigma+1} \leftarrow F_\sigma \cup \{\ell\}$ and charge the addition of $\ell$ to $u_{r-1}u_r,v_{r-1}v_r,u_{r+1}w\in APX_i$.
    
    To help ensure that $u_{r+1}w\in \delta(G_s)\cup APX_i$ is not charged to again, we consider if $w\in R$ or $w\in V(G_x)$ for some $G_x\in \{G_1,\dots, G_k\}$. If $w\in R$, we are done with this iteration, and $u_{r+1}w$ is in $Y_{i+1}$. If not, then there is some $G_x\in \{G_1,\dots, G_k\}$ such that $w\in V(G_x)$.

    If $w\in G_x$ such that every obstruction of $G_s$ is covered, then we are done with this iteration. So assume this is not the case, and let $S\subseteq  G_x$ be the square that contains $w$. if $S$ is covered by $F_{\sigma+1}$, we are done with this iteration. So assume that $S$ is not covered by $F_{\sigma+1}$. Since no link is necessary, there is a link $\ell'\in E(S)\cap L$ that is incident to $w$. We set $F_{\sigma+1}\leftarrow F_{\sigma+1}\cup \{\ell'\} \backslash \{u_{r+1}w\}$. So $F_{\sigma+1}$ now covers both $S_{r}$ and $S$ in addition to the obstructions previously covered, for no change in size. Thus, we do not need to find links in $APX_i$ to charge $\ell'$ to. We see that $G_x$ is in  $\cH_{i+1}$.    
    
    It remains to show that no link in $EC_3$ is charged to more than once. So assume for contradiction that link $\ell\in EC_3$ is charged to a second time in step $i+1$ when we are constructing $APX_{i+1}$. Let $G_s\in \{G_1,\dots, G_k\} \backslash \cH_{i}$ be the ladder that is being covered in step $i$. First, suppose $\ell\in E(G_S)$. By property (b), this implies that $G_s$ was covered in a previous iteration, and thus $G_s\in \cH_i$, a contradiction. 

    So suppose that $\ell\in \delta(G_s)$. By construction $\ell$ is charged to only when the incident square in $G_s$ is not covered by $APX_i$, and by $EC_3$. However, since $\ell$ has been charged to before, its other endpoint must be some $G_t\in \{G_1,\dots,G_k\}$ that has a square incident to $\ell$ that was uncovered before that iteration. However, in our construction, when we charge to $\ell\in \delta(G_t)$ the first time, we replace $\ell$ by a link in $G_s$ if the incident square in $G_s$ is uncovered. Therefore, $G_s\in \cH_i$, a contradiction.
\end{proof}
With Lemma~\ref{lem:uncovered} and Lemma~\ref{lem:simplifiedapproximation}, we have the ingredients necessary to prove Theorem~\ref{thm:unweighted}.
\begin{proof}[proof of Theorem~\ref{thm:unweighted}]
    Given an unweighted $(n-4)$-Node Connectivity Augmentation instance, we apply Theorem~\ref{thm:from_connectivity_to_obstructions} to obtain an instance $(G,4,L)$ of unweighted $d$-Obstruction Covering. We will show how to find a solution $APX$ with $|APX| \leq \frac{4}{3}|opt|$. Using Theorem~\ref{thm:ChainDecomposition}, we decompose $G$ into node-disjoint hexagons and ladders $G_1,\dots, G_k$, and lonely nodes $R$. We will compute the solution  $APX$ by adding partial solutions.

    By applying Lemma~\ref{lem:degree123}, we compute a minimum cardinality subset of links $EC_3\subseteq L$ that covers the obstructions contained in every subgraph $G_i$ of degree $1$ and $2$ and every degree $3$ node.  Let $\cH_0 \subseteq \{G_1, \dots, G_k\}$ denote the ladders and hexagons covered by $EC_3$. We set $APX_0 \coloneqq EC_3\cap E[\cH_0]$, and remove these edges from $G$ to get new $4$-Obstruction Covering instance $(G'=(V,E'),4,L')$, where $E'=E\backslash APX_0$, and $L' = L \backslash APX_0$. We apply Theorem~\ref{thm:ChainDecomposition} to find decomposition $R',G_1', \dots, G_{k_1}'$. Observe that by construction of $G'$, $G_1', \dots, G_{k_1}'$ contain only ladders of degree $3$ and $4$. 
    
    We apply Lemma~\ref{lem:uncovered} to $(G',4,L')$ with partial solution $EC_3\backslash APX_0$, to find a new partial solution $F_1$. We also find a subset of ladders $\cH_1\subseteq \{G_1',\dots, G_{k_1}'\}$ whose obstructions are covered by $F_1$, and a set of links $Y_1\subseteq F_1 \cap \delta(G'[\cH_1])$ where every  $y \in Y_1$ has an endpoint in $R'$. We set $APX_1 \coloneqq F_1 \cap(Y_1 \cup E'[\cH_1])$, and remove $APX_1$ from $G'$ to get new $4$-Obstruction Covering instance $(G'',4,L'')$ where $G''=(V',E'\backslash APX_1)$ and $L''  = L'\backslash APX_1$. We apply Theorem~\ref{thm:ChainDecomposition} to find decomposition $R'',G_1'', \dots, G_{k_2}''$. Observe that by Lemma~\ref{lem:uncovered},  $G_1'', \dots, G_{k_2}''$ contains only ladders of length $1$ and degree $3$ or $4$. 
    
    Lastly, we apply Lemma~\ref{lem:simplifiedapproximation} with $F_1\backslash APX_1$ as our partial solution, to find a feasible solution $APX_2$. We denote the ladders covered by $APX_2$ by $\cH_2 \coloneqq \{G_1'',\dots, G_{k_3}''\}$ .
    
    We define $APX \coloneqq APX_0\cup APX_1\cup APX_2$. By construction $APX_0$, $APX_1$, and $APX_2$ are pairwise disjoint. To see that $APX$ is a feasible solution to $(G,4,L)$, first observe that $\cH_0,\cH_1,$ and $\cH_2$ are also pairwise disjoint and $\cH_0\cup \cH_1 \cup \cH_2 = \{G_1, \dots, G_k\}$, so $APX$ covers every obstruction of $\{G_1, \dots, G_k\}$, in addition, $EC_3$ covers every obstruction in $R$, and $APX$ covers every obstruction covered by $EC_3$ so $APX$ is feasible.
    \begin{align*}
        APX =&  APX_0\cup APX_1\cup APX_2 = APX_0 \cup (F_1\cap (Y_1\cup E'[\cH_1]) ) \cup APX_2 \\
        =&  (EC_3\cap E[H_0]) \cup (F_1\cap (Y_1\cup E[\cH_1]) ) \cup APX_2  
    \end{align*}
    Therefore, $|APX|$ is equal to
    \begin{align*}
        =& |APX_0| + |F_1\cap (Y_1\cup E[\cH_1])| + |APX_2| \\
        \leq& |APX_0| + \frac{4}{3}|(EC_3\backslash APX_0)\cap (APX_1)| + \frac{4}{3}|F_1\backslash(APX_1)|\\
        =& |APX_0| + \frac{4}{3}|(EC_3\backslash APX_0)\cap (APX_1)| + \frac{4}{3}|(EC_3\backslash APX_0)\backslash(APX_1)|\\
        <& \frac{4}{3}|APX_0| +\frac{4}{3}|EC_3\backslash APX_0| = \frac{4}{3}|EC_3| \leq \frac{4}{3}|opt|
    \end{align*}
    Where the inequality above follows by application of the inequalities found  in Lemmas~\ref{lem:uncovered}, and~\ref{lem:simplifiedapproximation}. The second equality follows from the equality in Lemmas~\ref{lem:uncovered}. The final inequality above follows since $opt$ is a feasible solution to the $N$-edge cover problem that $EC_3$ solves.
\end{proof}

    \section{ Missing proofs of Section~\ref{sec:hardness}}
\label{sec:hardness-proofs}

Along this section, we will refer to the following special case of the 3-SAT problem, known as 3-SAT-4. 

\begin{definition}[3-SAT-4 problem] The 3-SAT-4 problem is the special case of 3-SAT where each variable appears exactly four times in the boolean formula $\I$. The corresponding optimization version of the problem is denoted as MAX-3-SAT-4. \end{definition}

The 3-SAT-4 problem is known to be NP-hard, and furthermore its optimization version is known to be APX-hard~\cite{berman2003approximation}.

\subsection{ Proof of Lemma~\ref{lem:3sat22hard}}
\label{sec:3sat22hard}

    Consider a 3-SAT-4 instance $\I_4$. Our goal is to construct a 3-SAT-(2,2) instance $\I_{2,2}$ in such a way that $\I_{2,2}$ is satisfiable if and only if $\I_4$ is satisfiable.
    
    We distinguish three types of variables in $\I_4$: the variables that occur exactly four times as positive or four times as negative literals, the variables that occur three times as positive literals or three times as negative literals, and the variables that occur twice as positive literals. We will construct the new instance $\I_{2,2}$ by initially setting it to be equal to $\I_4$ and then remove or change literals of the first two types in order to turn the formula into a 3-SAT-(2,2) instance.

    First, suppose a variable $x$ occurs as $4$ positive literals in $\I_{2,2}$ (the case of negative literals can be handled symmetrically). In this case, we can assume that $x$ is true and remove the clauses containing $x$ from $\I_{2,2}$. For the variables left in $\I_{2,2}$ that occur less than 4 times after the removal, we will add dummy clauses and variables to bring them back. For each such literal $y$ that was removed, we add dummy variables $d_{y,1}$ and $d_{y,2}$ and the clauses \[(y\lor d_{y,1} \lor d_{y,2}) \wedge (d_{y,1} \lor \bar d_{y,1}\lor \bar d_{y,1}) \wedge (d_{y,2} \lor \bar d_{y,2}\lor \bar d_{y,2}).\] We remark that, if more than one literal associated to the same variable were removed, we create different variables for them in the process. If the variables $d_{y,1}$ and $d_{y,2}$ are set to true, then the three clauses are satisfied regardless of the value of the variable corresponding to literal $y$, so we obtain an equivalent formula without variables occurring four times as positive literals.

    Suppose that $x$ occurs as $3$ positive literals in $\I_{2,2}$ (the case of negative literals can be handled symmetrically). We introduce new variables $y_x,z_x,w_x$ to $\I_{2,2}$, then select two arbitrary positive literals of $x$ and replace one with $y_x$ and the other with $z_x$. We then add the clauses
    \[(x\lor \bar y_x\lor  \bar y_x) \wedge (y_x\lor \bar z_x\lor  \bar z_x) \wedge (z_x\lor  \bar w_x\lor  \bar w_x) \wedge (w_x\lor  w_x\lor  \bar x).\]
    These new clauses are satisfied if and only if $x=y=z=w$. Moreover, it is not hard to see that now there are two positive literals and two negative literals for $x,y,z,w$. In conclusion, $\I_{2,2}$ is equivalent to $\I_4$ and a valid 3-SAT-(2,2) instance.

\subsection{ Proof of Proposition~\ref{prop:GILIfeasible}}
\label{sec:GILIfeasible}

First of all, since $G_{\I}$ is a $3$-regular graph, it does not contain $K_{1,4}$ as subgraph. Therefore, the only forbidden subgraph we need to show does not exist is $K_{2,3}$. Assume, for the sake of contradiction, that $F$ is a $K_{2,3}$ subgraph of $G_{\I}$. Let $v\in F$ be degree $3$ in $F$, we will distinguish two cases:

    \begin{itemize}
        \item If $v$ is a clause node in $G_{\I}$, say corresponding to clause $C$, then its neighbors are three variable nodes $x_i,y_j,z_k$ for $i,j,k\in \{1,3,4,6\}$, for some (possibly equal) variables $x,y,z$. If any two such variables are different, then their only common neighbor is $v$, and hence it cannot be part of a $K_{2,3}$. On the other hand, if the three neighbors of $v$ belong to the same subgraph $H_x$, by construction there are always two such variable nodes without common neighbors other than $v$ (either $x_1$ and $x_4$, or $x_3$ and $x_6$), hence again $v$ cannot be part of a $K_{2,3}$
        \item If $v$ is a variable node in some subgraph $H_x$, we consider two cases: if $v=x_1$ (the cases $x_3, x_4$ or $x_6$ are symmetric), its neighbors in $G_{\I}$ are $x_2$, $x_6$ and some clause node, but the only node connected to $x_2$ and $x_6$ is $x_5$, which is not connected to any clause node; hence, $x_1$ cannot be part of a $K_{2,3}$. If $v=x_2$ (the case of $x_5$ is symmetric), its neighbors in $G_{\I}$ are $x_1$, $x_5$ and $x_3$ but they do not have two neighbors in common. Hence, no variable node can be part of a $K_{2,3}$.
    \end{itemize}
    Thus, $(G_{\I},4,L_{\I})$ is a valid instance of the $4$-Obstruction Covering problem.

\subsection{ Proof of Theorem~\ref{thm:n-dobscovhard}}
\label{sec:n-dobscovhard}
Given a $3$-SAT-$(2,2)$ instance $\I$, we first create the $(n-4)$-Obstruction Covering instance $(G_{\I},4,L_{\I})$ as described in the proof of Theorem~\ref{thm:n-4hard}. We will modify it to obtain an $d$-Obstruction Covering instance $(G_{\I}',d,L_{\I}')$, and show that finding an optimal covering of the obstructions in $G_{\I}'$ implies whether a satisfying assignment exists for $\I$ or not. 

Recall that the instance $(G_{\I},4,L_{\I})$ has a node for every clause $C$, and a subgraph $H_x$ for every variable $x\in\I$. For every clause node $C$, we add $d-4$ dummy nodes to $G_{\I}'$ that are adjacent to $C$, so that $C$ is degree $d-1$. For subgraph $H_x$, with nodes $x_1,\dots,x_6$, we add $d-4$ dummy nodes to $G_{\I}'$ that are adjacent to $x_3$ and $d-4$ dummy nodes adjacent to $x_6$. We add $d-4$ nodes to $G'$ that are adjacent to both $x_1$ and $x_5$, which we denote by $y_1,\dots,y_{d-4}$, and we add $d-4$ nodes that are adjacent to both $x_2$ and $x_4$, which we denote by $z_1,\dots, z_{d-4}$. None of these new edges are added to $L_{\I}'$, so in fact we have $L_{\I}'= L_{\I}$. 
    
Notice that every clause and variable node from $G_{\I}$ now is degree $d-1$, and thus is part of a $K_{1,d-1}$ obstruction. Moreover, the nodes $x_1, x_2, x_5, x_6, y_1,$ $\dots,y_{d-4}$ form a $K_{2,d-2}$ obstruction, as do the nodes  $x_2, x_3, x_4, x_5, z_1, \dots, z_{d-4}$. See Figure~\ref{fig:SAT-obstruction}~(Right). Since no set of at least three nodes has more than two neighbours in common in $G_{\I}'$, these are the only obstructions in the instance.

We will prove now that $(G_{\I}',d,L_{\I}')$ is a valid $n-d$ Obstruction Covering instance. Recall that the forbidden subgraphs of a $d$-Obstruction Covering instance is a $K_{i,j}$ subgraph where $i+j > d$. First, the dummy nodes adjacent to clause nodes are degree $1$, and the clause node are degree $d-1$, so these dummy nodes cannot be part of a forbidden subgraph. Similarly, for variables $x\in \I$, the dummy nodes adjacent to $x_3$ and $x_6$ are not part of a forbidden subgraph. The nodes $y_1, \dots, y_{d-4}$ are degree $2$ and thus cannot be part of a forbidden subgraph $K_{i,j}$ for $i>2$, and the same holds for $z_1, \dots, z_{d-4}$. Since the maximum degree of $G_{\I}'$ is $d-1$, these nodes cannot be part of a $K_{1,j}$ forbidden subgraph. Thus, they can only possibly be part of a $K_{2,j}$ forbidden subgraph. The neighbours of $y_1, \dots, y_{d-4}$ are $x_1$ and $x_5$ which also share a neighbourhood of $x_2$ and $x_6$, which together form a $K_{2,d-2}$ subgraph. In conclusion, the are no forbidden subgraphs in $G_{\I}'$, and hence $(G_{\I}',d,L_{\I})$ is a valid $d$-Obstruction Covering instance.

Since $L_{\I}'=L_{\I}$ and given the structure of the obstructions, $F\subseteq L_{\I}' = L_{\I}$ is a feasible solution to $(G_{\I}',d,L_{\I}')$ if only if $F$ is a feasible solution to $(G_\I,4,L_\I)$, and thus by Theorem~\ref{thm:n-4hard} the result holds.

As a final remark, notice that, if the 3-SAT-(2,2) instance $\I$ has $t$ clauses (and consequently $\frac{3}{4}t$ variables), the size of $G_{\I}'$ is $O(t\cdot d)$. In order for this procedure to be a polynomial-time reduction, it is required that $d \in O(n^{1-c})$ for some constant $c>0$, as otherwise $t$ must be subpolynomial with respect to $n$, leading to superpolynomial time for the reduction algorithm.

\subsection{ Proof of Theorem~\ref{cor:n-dobscovapxhard}}
\label{sec:obscovapxhard}

We will first prove that MAX-3-SAT-(2,2) is APX-hard, and then use this result to conclude that the $4$-Obstruction Covering problem is APX-hard as well; the proof of Theorem~\ref{cor:n-dobscovapxhard} will be a direct consequence of that by construction.

\begin{lemma}\label{lem:3SAT22apxhard} 
    The MAX-3-SAT-(2,2) problem is APX-hard. 
\end{lemma}

\begin{proof} 
    Consider an instance $\I$ of MAX-3-SAT-4 having $C$ clauses, and being $OPT_{\I} = C-t$ the number of clauses satisfied by an optimal solution, where $t\in\mathbb{N}$. We apply the same procedure described in Lemma~\ref{lem:3sat22hard} to construct an equivalent instance $\I'$ of MAX-3-SAT-(2,2), and let $OPT_{\I'}$ be the number of clauses satisfied by an optimal solution.

    First of all, notice that $\I'$ has clauses that come from $\I$ (possibly with some auxiliary variables instead of the original) plus some auxiliary clauses added in the construction. In particular, if we fix the values of the auxiliary variables so that $x=y_x=z_x=w_x$ for each $x$ where the corresponding modification was added, and also set the variables $d_{y,1}$ and $d_{y,2}$ added in the process to be true, then all the auxiliary clauses are satisfied. This means that $OPT_{\I'} \ge C' - t$, where $C'$ is the number of clauses in $\I'$. Furthermore, we have that $C'\le 12 C$: The number of literals in $C$ is $3C$, and for each literal we include up to four extra clauses (if they were removed and had to be included back we include three, and then we may include four extra clauses per variable).
    
    Suppose now that there is a PTAS for MAX-3-SAT-(2,2). We will run this PTAS on instance $\I'$, obtaining an assignment of values that satisfies at least $(1-\varepsilon)OPT_{\I'}$ clauses. Let us derive an assignment of values for $\I$ by assigning the same values to the variables that are in both $\I$ and $\I'$, and assigning value true to the variables that occur as four positive literals in $\I$ and false to the variables that occur as four negative literals in $\I$. In the worst case, all the unsatisfied clauses in $\I'$ also appear in $\I$, which implies that the derived assignment satisfies at least \[C - t - \varepsilon(C' - t) \ge C - t -\varepsilon(12C-t)\] clauses of $\I$.
    
    Notice that $C-t \ge \frac{C}{2}$, because if we consider the assignment that sets every variable to true and the assignment that sets every variable to false, one of them must satisfy at least half of the clauses. 
    Putting everything together, we have that the number of clauses satisfied by our derived assignment for $\I$ is at least \[C - t -\varepsilon(12C-t) \ge C - t - 24\varepsilon(C-t) + \varepsilon t \ge (1-24\varepsilon)OPT_{\I},\] implying that there exists a PTAS for MAX-3-SAT-4, \textcolor{black}{a contradiction}. 
\end{proof}

\begin{theorem}
\label{thm:obscovapxhard} 
    The $4$-Obstruction Covering problem is APX-hard.
\end{theorem}

\begin{proof} Consider a MAX-3-SAT-(2,2) instance $\I$ with optimal solution $OPT_{\I} = m-t$, where $m$ is the number of clauses and $t\in \mathbb{N}$. Consider the $4$-Obstruction Covering instance $(G_{\I},4, L_{\I})$ as defined in the proof of Theorem~\ref{thm:n-4hard}. Observe that there exists a feasible solution to $(G_{\I},4,L_{\I})$ whose size is $4k+t$, where $k$ is the number of variables in $\I$: For every variable $x$, if it is set to true in the optimal assignment for $\I$, then we take the links corresponding to the positive literals of the subgraph $H_x$ that connect them to their corresponding clauses plus $x_2x_3$ and $x_5x_6$, and otherwise we take the links corresponding to the negative literals instead plus $x_1x_2$ and $x_4x_5$. This choice of links will cover every obstruction in the subgraphs $H_x$, and the nodes corresponding to clauses that are satisfied by the optimal assignment in $\I$. For the clause nodes corresponding to clauses that are not satisfied, we simply take an arbitrary link incident to them. As there are $k$ subgraphs $H_x$, and their obstructions are covered by four links each, this solution has a cost of $4k+t$.



Let us assume that there exists a PTAS for the $4$-Obstruction Covering problem. We will run the PTAS on $(G_{\I},4,L_{\I})$, obtaining a solution of size at most $(1+\varepsilon)OPT_{(G_{\I},4,L_{\I})} \le (1+\varepsilon)(4k+t)$, where $OPT_{(G_{\I},4,L_{\I})}$ is the size of the optimal solution for $(G_{\I},4,L_{\I})$. We will construct an assignment of values for the variables in $\I$ as follows: For each variable $x$, if only the links connecting the variable nodes corresponding to positive literals of $x$ to their clauses are chosen, then we set the value of $x$ to be true; if only the links connecting the variable nodes corresponding to negative literals of $x$ to their clauses are chosen, then we set the value of $x$ to be false; otherwise, we set the value of $x$ to make true the literal with most incident chosen links that connect them to clause nodes (breaking ties arbitrarily). As argued in the proof of Lemma~\ref{lem:3sat22ton-4}, if for some variable $x$ we pick links connecting clause nodes to variable nodes which are not in opposite positions (i.e. corresponding to both positive and negative literals of $x$), then we require at least five links to cover the obstructions appearing in such a subgraph $H_x$; furthermore, if four links connecting clause nodes to $H_x$ are chosen, then since $x_2$ and $x_5$ are yet to be covered and $x_2x_5$ is not a link, we need at least six links to cover those obstructions. Thus, at most $t + \varepsilon(4k+t)$ clauses of $\I$ are not satisfied.

Similarly to the proof of Lemma~\ref{lem:3SAT22apxhard}, we have that $m-t \ge \frac{m}{2}$. Putting everything together, we get that the number of satisfied clauses in $\I$ by the constructed assignment is at least \[m-t - \varepsilon(4k+t) = m-t - \varepsilon(3m-t) \ge m-t - 6\varepsilon(m-t) + \varepsilon t \ge (1-6\varepsilon)OPT_{\I},\] implying that there is a PTAS for MAX-3-SAT-(2,2). \end{proof}


Now we can prove Theorem~\ref{cor:n-dobscovapxhard}. Consider the same construction as in the proof of Theorem~\ref{thm:n-dobscovhard}. Since, by construction, the set of links are the same for both instances and the obstructions are in a $1$-to-$1$ correspondence, the reduction is directly approximation-preserving, hence proving the result.
\end{appendix}
\end{document}